\documentclass{article}

\usepackage[utf8]{inputenc}
\usepackage[retainorgcmds]{IEEEtrantools}
\usepackage{amsfonts,latexsym,amssymb}
\usepackage{mathrsfs,amsmath}
\usepackage{amsthm}
\usepackage{bbm}
\usepackage{subcaption}
\usepackage{upgreek}
\usepackage[toc,page]{appendix}
\usepackage{array}
\usepackage{epstopdf,epsfig}

\newtheorem{theorem}{\bf Theorem}[section]
\newtheorem{prop}{\bf Proposition}[section]

\newtheorem*{theorem*}{\bf Theorem \ref{th:global_eta}}
\newtheorem*{remark}{\bf Remark}

\graphicspath{ {./Images/} }

\newcommand{\vecD}{\mathbf{D}}
\newcommand{\vecE}{\mathbf{E}}
\newcommand{\vecF}{\mathbf{F}}
\newcommand{\vecG}{\mathbf{G}}

\newcommand{\vecI}{\mathbf{I}}
\newcommand{\vecJ}{\mathbf{J}}

\newcommand{\vecS}{\mathbf{S}}
\newcommand{\vecT}{\mathbf{T}}

\newcommand{\vecV}{\mathbf{V}}
\newcommand{\vecW}{\mathbf{W}}

\newcommand{\vep}{\varepsilon}
\newcommand{\vecb}{\mathbf{b}}

\newcommand{\vecf}{\mathbf{f}}

\newcommand{\vecn}{\mathbf{n}}

\newcommand{\vecq}{\mathbf{q}}
\newcommand{\vecr}{\mathbf{r}}
\newcommand{\vecs}{\mathbf{s}}
\newcommand{\vect}{\mathbf{t}}

\newcommand{\vecv}{\mathbf{v}}

\newcommand{\vecz}{\mathbf{z}}

\newcommand{\vecphi}{\boldsymbol\upphi}

\newcommand{\vectau}{\boldsymbol\tau}

\newcommand{\thetad}{\dot{\theta}}
\newcommand{\phid}{\dot{\phi}}
\newcommand{\vecphid}{\dot{\boldsymbol\upphi}}

\newcommand{\vecqd}{\dot{\mathbf{q}}}

\newcommand{\Rsv}{\mbox{\boldmath$\mathcal{R}$}}
\newcommand{\Gc}{\mbox{\boldmath$\mathcal{G}$}}

\newcommand{\ct}{\mathrm{c_t}}

\newcommand{ \vctwo }[2] {
 \left[\begin{array}{c }
  #1        \\
  #2
\end{array} \right]  }

\newcommand{ \matwo }[4] {
 \left[\begin{array}{c c }
   #1  &  #2     \\
   #3  &  #4
\end{array} \right]  }

\newcommand{ \matwothree }[6] {
 \left[\begin{array}{c c}
   #1  &  #2   \\
   #3  &  #4   \\
   #5  &  #6
\end{array} \right]  }

\newcommand{ \mathree }[9] {
 \left[\begin{array}{c c c}
   #1  &  #2  &  #3   \\
   #4  &  #5  &  #6   \\
   #7  &  #8  &  #9
\end{array} \right]  }

\title{Optimization and small-amplitude analysis of Purcell's three-link microswimmer model}

\author{Oren Wiezel and Yizhar Or
\thanks{O. Wiezel and Y. Or are with the Faculty of Mechanical Engineering, Technion - Israel Institute of Technology, Haifa 32000, Israel {\tt\small izi@technion.ac.il}}%
}

\begin{document}

\maketitle
\begin{abstract}
This work studies the motion of Purcell's three-link microswimmer in viscous flow, by using perturbation expansion of its dynamics under small-amplitude strokes. Explicit leading-order expressions and next-order correction terms for the displacement of the swimmer are obtained for the cases of a square or circular gait in the plane of joint angles. The correction terms demonstrate the reversal in movement direction for large stroke amplitudes, which has previously only been shown numerically. In addition, asymptotic expressions for Lighthill's energetic efficiency are obtained for both gaits. These approximations enable calculating optimal stroke amplitudes and swimmer's geometry (i.e. ratio of links' lengths) for maximizing  either net displacement or Lighthill's efficiency.
\end{abstract}


\section{Introduction}

The study of micron-size swimmers dynamics has in recent years become a highly active research area and has possible implications on understanding the motion of swimming microorganisms and biological infections \cite{LaugaPowers09,lauga2015bacterial}. The considerable advances made in the field of micro and nano-technology have promoted the possibility of manufacturing miniature robotic devices that operate in these small scales\cite{abbott2009should,kosa2008flagellar,jang2015undulatory}. Such mini-robots may have many applications in medicine, performing medical procedures in a minimally invasive way and delivering drugs with high precision \cite{gao2012cargo,PeyerNelson2013,sitti2015biomedical}. All this requires an understanding of swimming dynamics at the low Reynolds number regime.

Reynolds number represents the ratio of the inertial forces to the viscous ones. When dealing in microfluidics, where the characteristic lengths are extremely small, hydrodynamic forces are typically governed by very low Reynolds numbers ($\text{Re}\ll1$) \cite{happel&brenner_book}. 
The result of this is that the strategy of motion in this regime needs to be drastically different than the familiar motion of larger organisms, such as fish, that rely on imparting momentum to the surrounding fluid.
In his famous lecture \cite{purcell1977life}, Purcell introduced the "Scallop theorem" which states that a swimmer  that changes its shape and then changes back to the original shape by reversing the same sequence will return to the point it started its motion with no regard to the speed in which any part of the motion is made.  Any reciprocal motion of the swimmer will result in zero net translation.
There are several ways of overcoming the scallop theorem and generating net motion in a low Reynolds regime. 
One of them is continuously performing a unidirectional rigid body motion such as a rotating corkscrew. In this way there is no reciprocal motion and the result is generation of net motion of the swimmer. This is precisely the method used by the "Escherichia coli" bacteria to propel itself \cite{berg1973bacteria}.
Another way for overcoming the scallop theorem is by making a non-reciprocal periodic shape change, which will be henceforth called a "gait". The simplest version of such a swimmer, known as "Purcell's 3-link swimmer" (Figure \ref{fig:swimmer}), was suggested by Purcell in \cite{purcell1977life} and can be seen as a simplified version of the travelling wave "Taylor sheet" \cite{taylor1951analysis} in two dimensions which is discretized to have only two degrees of freedom.
Purcell's swimmer is comprised of three rigid links connected by two rotary joints (see figure \ref{fig:swimmer}).
 \begin{figure}[!b]
    \centering
    \includegraphics[width=0.5\textwidth]{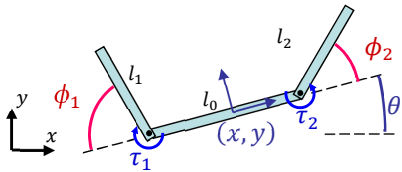}
    \caption{The "Purcell" 3-link swimmer}
    \label{fig:swimmer}
\end{figure}
Purcell indicated that this swimmer could propel itself along a straight line by alternately rotating its front and back links in a non-reciprocal way (Square gait. Figure \ref{fig:gaits.sq}). Through symmetry considerations alone, it can be shown that the three link swimmer will move along the $x$ axis when using the shape changes suggested by Purcell.
\begin{figure}[!t]
    \centering
    \begin{subfigure}{0.45\textwidth}
    \includegraphics[width=\textwidth]{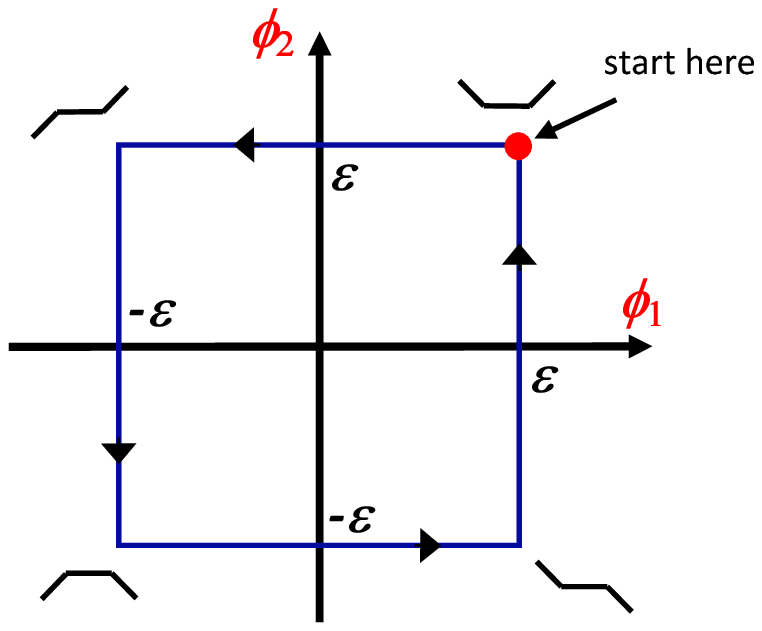}
    \caption{}
    \label{fig:gaits.sq}
    \end{subfigure}
    \begin{subfigure}{0.45\textwidth}
    \includegraphics[width=\textwidth]{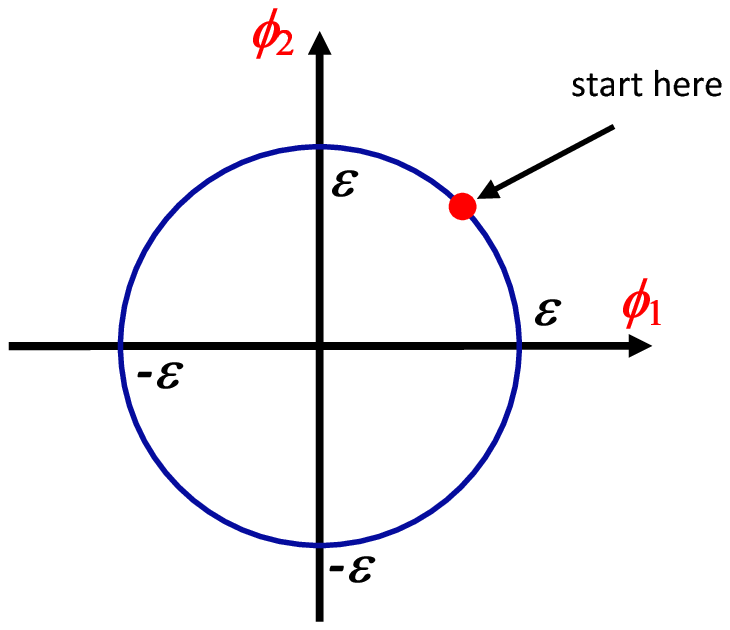}
    \caption{}
    \label{fig:gaits.cir}
    \end{subfigure}
    \caption{Gait representation in joint angle plane. (a) square gait (b) circular gait.}
    \label{fig:gaits}\vspace{-10pt}
\end{figure}
Purcell claimed that determining the direction of net motion (i.e. forward or backward) is trivial and left it as "an exercise for the student".
Only 26 years later, Becker et al \cite{becker2003self} obtained an explicit formulation for the dynamics of this microswimmer, and surprisingly found that the direction of net motion actually depends on the stroke amplitude of joint angles. Specifically, for small amplitudes the swimmer will move in one direction but for larger amplitudes the swimmer will move in the opposite direction. Additionally, \cite{becker2003self} studied Lighthill's energetic efficiency, which is roughly equivalent to the maximal mean speed achievable under a given mean expenditure of mechanical power. 
It is worth noting that both \cite{purcell1977life} and \cite{becker2003self} considered only the case where the joint angles rotate alternately, creating a square gait, and did not study other possible periodic shape change, such as the circular gait (Figure \ref{fig:gaits.cir}), at which the two joint angles oscillate sinusoidally with a quarter-period phase shift.

Some previous works have examined Purcell's swimmer from different prospectives. In \cite{Gutman2015Symmetries} Gutman and Or analyze the symmetries of Purcell's swimmer, derive conditions on gaits that result in movement along the swimmer's principal directions, and present motion experiments of a macro-scale robotic swimmer in a highly viscous fluid. Avron and Raz \cite{avron&raz08} and Hatton et al \cite{hatton2013geometric,hatton2013tro} utilize tools of differential geometry in order to obtain geometric visualization of performance measures in shape space (plane). By observing curvature maps of the swimmer's dynamics, those works can qualitatively explain the reversal in direction of motion for large-amplitude strokes, and also provide guidelines that help in obtaining optimal gaits. Tam and Hosoi \cite{tam2007optimal}, through numerical computation only, found gaits that achieve optimal translation during a period as well as energetically optimal gaits. They also found the {\em optimal geometry} of the swimmer, i. e. ratio between the swimmer's links lengths, for these two optimality criteria, again only through numeric calculations.

Aside from numerical computations, another useful approach is obtaining closed-form expressions under some simplifying scaling assumptions, by using asymptotic analysis.

In microswimmers, this approach dates back to the analysis of Taylor \cite{taylor1951analysis}, who showed that the swimming speed of Taylor's sheet scales quadratically with the wave's amplitude at leading order. In \cite{shapere1989geometry}, this concept was employed on analysis of self-propulsion of spherical swimmers performing a squirming motion. For the shape-changing three-link swimmer, one has to utilize the method of perturbation expansion \cite{nayfeh2008perturbation} under the assumption of small stroke amplitudes. This method was used in \cite{becker2003self} for obtaining leading-order expressions for the motion of Purcell's swimmer and finding its optimal geometry. In the work \cite{giraldi2015optimal}, leading order expressions were derived using a similar approach and then used to find optimal geometry of the swimmer for the square gait, and for derivation of an optimal polygonal gait under small joint angles. Interestingly, there are differences in the results found in \cite{becker2003self}, \cite{tam2007optimal} and \cite{giraldi2015optimal} in terms of the swimmer's optimal geometry.
Recently, perturbation expansion has become a more common tool for analysis of simple microswimmers motion in various cases such as a three-link swimmer with a passive elastic joint \cite{passov2012dynamics}, a two-link magnetically actuated microswimmer \cite{gutman2014simple}, and wobbling of a magnetically actuated helix \cite{man2013wobbling}.

The goal of this work is to introduce a systematic method for analyzing the dynamics of Purcell's three link swimmer using perturbation expansion and to exploit this method to perform optimizations on the swimmer geometry and stroke amplitude.
The main new contributions of this work compared to previous literature are: a. Formulation of leading-order expression for the displacement under the circular gait (Figure \ref{fig:gaits.cir}). b. Derivation of the next-order corrections terms for the displacement under both square and circular gaits, which explicitly show the reversal in movement direction first found in \cite{becker2003self}, as well as obtaining an approximation for the optimal stroke amplitude. c. Obtaining leading-order expressions for Lighthill's energetic efficiency of the two presented gaits and also a next-order term for the square gait, as well as using these expressions to perform optimization of the swimmer's geometry and stroke amplitude for maximal efficiency. d. Resolving the previously unexplained differences between the findings in \cite{giraldi2015optimal},\cite{tam2007optimal} and \cite{becker2003self} regarding the optimal geometry.

The paper is organized as follows.
The next section presents Purcell's three-link swimmer model and its equations of motion, the two gaits in question and the concept of Lighthill's energetic efficiency. Section 3 focuses on perturbation expansion, and presents a systematic way for deriving the net displacement of the swimmer in the form of a power series for any gait, and the first two terms are found explicitly for the square and circular ones. The optimal geometry and stroke amplitude for maximal displacement are also found. This section also includes the use of perturbation expansion to find the period time under constant joint torque for the square gait, and under constant power for both square and circular gaits.
In section 4 the previous results are used in order to obtain the leading-order terms of the energetic efficiency of the two gaits, as well as next-order term for the square gait. Optimal geometry and amplitude are then found for maximal efficiency.
Finally, section 5 offers our conclusions, indicating possible future extensions and consequences.

\section{Problem formulation}
In this section we introduce Purcell's swimmer model, formulate its dynamic equations of motion, define the square and circular gaits and finally, review the concept of Lighthill's energetic efficiency.
The swimmer in question consists of three thin rigid links with lengths $l_0,l_1,l_2$, with $l_1=l_2$ , $l=l_0+l_1+l_2$ and $\eta=\frac{l_0}{l}$. The links are connected by two rotary joints whose angles are denoted by $\phi_1$ and $\phi_2$ (see figure \ref{fig:swimmer}). The shape of the swimmer will be described by these two angles $\vecphi=(\phi_1,\phi_2)^T$. It is assumed that the swimmer's motion is confined to $x-y$ plane. The planar position and orientation of the center link are denoted by $\vecq=(x,y,\theta)^T$ . The swimmer is submerged in an unbounded fluid domain whose motion is governed by Stokes equations under low Reynolds number \cite{happel&brenner_book}.
The velocity of the $i^{th}$ link is described by the linear velocity of its center $\vecv_i$ and the link's angular velocity $\omega_i$ , which are augmented in the vector $\vecV_i=(\vecv_i,\omega_i)\in\mathbb{R}^3$. Similarly, let $\vecV_b=(\vecv_b,\omega_b)$ denote the linear and angular velocities of a body-fixed reference frame. In addition, we define the indicator matrix $I_{ij}$, such that $I_{ij}=\pm 1$ if the location of the $i^{th}$ link with respect to the body frame is affected by the $j^{th}$ joint, with the sign determined by the direction of the joint axis. $I_{ij}=0$ if the $j^{th}$ joint does not affect the $i^{th}$ link. The internal torques (moments) acting on the joints are denoted by $\tau_1$ and $\tau_2$.

\subsection{Formulation of equations of motion}
The kinematic relation between body velocity, joints velocities and links velocities is given by:
\begin{IEEEeqnarray}{c}
\vecv_i=\vecv_b+\omega_b\vecz\times(\vecr_i-\vecr_0)+\sum_j I_{ij}\phid_j\vecz\times(\vecr_i-\vecb_j)\nonumber\\[12pt]
\omega_i=\omega_b+\sum_j I_{ij}\phid_j
\end{IEEEeqnarray}
where $\vecr_i$ is the position of the center of the $i^{th}$ link and $\vecb_j$ is the position of the $j^{th}$ joint.
In matrix form, the velocity $\vecV_i$  is related to the body velocity $\vecV_b$  and shape velocity $\vecphid$  through:
\begin{equation}
 \vecV_i=\vecT_i(\vecq,\vecphi) \dot{\vecq}+\vecE_i(\vecq,\vecphi) \vecphid.
 \label{eq:kin}
\end{equation}
The matrices $\vecT_i(\vecq,\vecphi)$  and $\vecE_i(\vecq,\vecphi)$  for $i=0,1,2$  are given by:
\begin{equation}
\begin{array}{l}
\vecT_0=\mathree{1}{0}{0}{0}{1}{0}{0}{0}{1},\; \vecE_0=\matwothree{0}{0}{0}{0}{0}{0} \\[16pt]
\vecT_1 = \mathree{1}{0}{-0.5 l_o\sin\alpha_0 - 0.5 l_1\sin\alpha_1}{0}{1}{0.5 l_0\cos\alpha_0 + 0.5 l_1\cos\alpha_1}{0}{0}{1}, \vecE_1 = \matwothree{-0.5 l_1\sin\alpha_1}{0}{0.5 l_1\cos\alpha_1}{0}{1}{0}\\[16pt]
\vecT_2 = \mathree{1}{0}{0.5 l_o\sin\alpha_o + 0.5 l_2\sin\alpha_2}{0}{1}{-0.5 l_o\cos\alpha_o - 0.5 l_2\cos\alpha_2}{0}{0}{1}, \vecE_2 = \matwothree{0}{-0.5 l_2\sin\alpha_2}{0}{0.5 l_2\cos\alpha_2}{0}{-1},
\end{array}
\end{equation}
Where $\alpha_0=\theta$, $\alpha_1=\phi_1-\theta$, and $\alpha_2=\theta+\phi_2$ are the absolute orientation angles of each link. Next, we invoke {\em resistive force theory } \cite{cox1970,gray1955propulsion}, which states that the viscous drag force $\vecf_i$  and torque $m_i$  on the $i^{th}$ slender link under planar motion are proportional to its linear and angular velocities. Thus, we can write the expression for the drag force and torque exerted on each link:
\begin{equation}
\begin{array}{c}
\vecf_i=-c_t l_i(\vecv_i\cdot\vect_i)\vect_i-c_n l_i(\vecv_i\cdot\vecn_i)\vecn_i \\[12pt]
m_i=-\dfrac{1}{12} c_n l_i^3 \omega_i,
\end{array}
\label{eq:RFT}
\end{equation}
Where $\vect_i=(\cos\alpha_i,\sin\alpha_i)^T$ is a unit vector in the axial direction of the $i^{th}$ link and $\vecn_i=(-\sin\alpha_i,\cos\alpha_i)^T$  is a unit vector in the normal direction. The resistance coefficients for the normal and axial directions are $c_n\!=\!2\ct\!=4\pi\mu/log(l_c/a)$ where $\mu$  is the dynamic viscosity of the fluid, $a$  is the radius of the links and $l_c$ is a characteristic length. Since the ratio of the link's length to its radius is assumed to be very large, the difference in the resistance coefficients between the links is very small and will be neglected. It is also assumed that effects of hydrodynamic interaction between the links are negligible.
Denoting the vector of forces and torques on the $i^{th}$ link as $\vecF_i=(\vecf_i,m_i)$ the relation (\ref{eq:RFT}) can be written in matrix form
\begin{equation}
\vecF_i=-\mathcal{R}_i(\vecq,\vecphi)\vecV_i
\label{eq:matrixrft}
\end{equation}
 where:
\begin{equation}
\mathcal{R}_i(\vecq,\vecphi)=\ct^{(i)}l_i \mathree{1+\sin^2\alpha_i}{-\cos\alpha_i\sin\alpha_i}{0}{-\cos\alpha_i\sin\alpha_i}{1+\cos^2\alpha_i}{0}{0}{0}{\dfrac{1}{6}l_i^2}
\end{equation}
is called the resistance tensor. The total drag forces and torques acting on the swimmer's body are given by
\begin{equation}
\vecf_b=\sum_{i=0}^2 \vecf_i, \quad m_b=\sum_{i=0}^2 \left(m_i+\left(\left(\vecr_b-\vecr_i\right) \times \vecf_i\right)\cdot\vecz\right).
\label{eq:fbmb}
\end{equation}
Using the matrices $\vecT_i$ from the kinematic relation (\ref{eq:kin}) and augmenting in a vector $\vecF_b=(\vecf_b,m_b)$, \eqref{eq:fbmb} is written in matrix form as:
\begin{equation}
\vecF_b=\sum_{i=0}^2\vecT_i^T\vecF_i =-\sum_{i=0}^2\vecT_i^T\Rsv_i(\vecT_i \vecV_b +\vecE_i\vecphid).
\label{eq:forces}
\end{equation}
Denoting
\begin{equation}
\mathcal{R}_{bb}=\sum_{i=0}^2\vecT_i^T\mathcal{R}_i\vecT_i \quad \text{and} \quad \mathcal{R}_{bu}=\sum_{i=0}^2\vecT_i^T\mathcal{R}_i\vecE_i,
\label{eq:Rbb@Rbu}
\end{equation}
 equation (\ref{eq:forces}) becomes:

\begin{equation}
\label{eq:Fb}
\vecF_b=\mathcal{R}_{bb} \vecV_b +\mathcal{R}_{bu}\vecphid.
\end{equation}
Our choice of coordinates of body location $\vecq$ implies that the body velocity satisfies $\vecqd= \vecV_b$. Assuming quasi-static motion, the swimmer is in static equilibrium $\vecF_b=0$. Substituting this into \eqref{eq:Fb} then yields the nonlinear differential equations that govern the swimmer's dynamics:
\begin{equation}
\vecqd=\Gc(\vecq,\vecphi)\vecphid
\label{eq:dynamics}
\end{equation}
where
\begin{equation}
\Gc(\vecq,\vecphi)=-\mathcal{R}_{bb}^{-1}
\mathcal{R}_{bu}
\label{eq:GandR}
\end{equation}
Since there are no external boundary conditions in the unbounded fluid domain, the swimmer's velocity expressed in body-fixed reference frame is independent of the position $\vecq$ and thus  \eqref{eq:dynamics} can be written as (cf. \cite{Gutman2015Symmetries}):
\begin{equation}\label{eq:gauge_sym}
\begin{array}{l}
\vecqd=\vecD(\theta)\vecG(\vecphi)\vecphid, \\[15pt]
\mbox{where } \vecD(\theta) = \mathree{\cos\theta}{-\sin\theta}{0}{\sin\theta}{\cos\theta}{0}{0}{0}{1}
\end{array}
\end{equation}
Note that \eqref{eq:gauge_sym} implies that the angular velocity of the swimmer $\thetad$  is independent of $\theta$.
The matrix $\vecG(\vecphi)$ obeys some symmetry relations due to the swimmer's structure (see \cite{Gutman2015Symmetries} for details):
\begin{IEEEeqnarray}{rCl}
\vecG(-\vecphi)&=&\mathree{-1}{0}{0}{0}{1}{0}{0}{0}{1}\vecG(\vecphi) )\nonumber\\
\vecG(\vecS\vecphi)&
=&\mathree{-1}{0}{0}{0}{1}{0}{0}{0}{-1}\vecG(\vecphi)\vecS
\end{IEEEeqnarray}
Where $\vecS=\matwo{0}{1}{1}{0}$  represents interchanging between the joint angles $\phi_1,\phi_2$  .

Next, we compute the actuation torques acting at the joints,
which are denoted by the vector $\vectau=(\tau_1,\tau_2)^T$. Due to static equilibrium of the partial kinematic chain ending at the $j^{th}$ joint, these torques are balanced by hydrodynamic forces and torques $\vecf_i$ and $m_i$, giving rise to the following relation:
\begin{equation}
\tau_j=-\sum_iI_{ij}\left( m_i+\left( \left(\vecr_i-\vecb_j\right) \times\vecf_i\right)\cdot \vecz \right)
\end{equation}
This relation can be written in matrix form using the previously defined matrices $\vecE_i$:
\begin{equation}
\vectau=-\sum_{i=0}^2\vecE_i^T\vecF_i
\label{eq:torque}
\end{equation}
Substituting (\ref{eq:kin}), (\ref{eq:matrixrft}) and (\ref{eq:dynamics}) in (\ref{eq:torque}) gives:
\begin{IEEEeqnarray}{rCl}
\vectau &=& \sum_{i=0}^2\vecE_i^T\mathcal{R}_i\vecV_i\nonumber\\
		&=& \sum_{i=0}^2\vecE_i^T\mathcal{R}_i\left(\vecT_i\vecqd_i+\vecE_i\vecphid_i\right)\nonumber\\
		&=&
\left(\sum_{i=0}^2\vecE_i^T\mathcal{R}_i\vecE_i-\sum_{i=0}^2\vecE_i^T\mathcal{R}_i\vecT_i\left(\sum_{i=0}^2\vecT_i^T\mathcal{R}_i\vecT_i\right)^{-1}\sum_{i=0}^2\vecT_i^T\mathcal{R}_i\vecE_i\right)\vecphid
\end{IEEEeqnarray}
denoting $\mathcal{R}_{uu}=\sum_{i=0}^2\vecE_i^T\mathcal{R}_i\vecE_i$ this equation can be written as:
\begin{equation}
	\vectau=\left(\mathcal{R}_{uu}-\mathcal{R}_{bu}^T\mathcal{R}_{bb}^{-1}\mathcal{R}_{bu}\right) \vecphid
\end{equation}
with $\mathcal{R}_{bb}$ and $\mathcal{R}_{bu}$ as defined earlier in (\ref{eq:Rbb@Rbu}).
Denoting
\begin{equation}
\vecW(\vecphi)=\mathcal{R}_{uu}-\mathcal{R}_{bu}^T\mathcal{R}_{bb}^{-1}\mathcal{R}_{bu}
\label{eq:W(phi)}
\end{equation}
 we have:
\begin{equation} \label{eq:tau_vs_phid}
	\vectau=\vecW(\vecphi)\vecphid.
\end{equation}
The resulting dynamic equations in \eqref{eq:dynamics} and \eqref{eq:tau_vs_phid} are strongly nonlinear. Nevertheless, they can be expanded by assuming small-amplitude changes of the angles about $\vecphi=0$, as explained in Section 3.

\subsection{Periodic Gaits}
This work considers time-periodic inputs of shape changes $\vecphi(t)$, called {\em gaits}, which represent closed loops in the plane of joint angles. In particular, we focus on two specific possible gaits, square and circular (Figure \ref{fig:gaits}), that are presented here with $\varepsilon$ as a scaling factor of the stroke amplitude, which will later be assumed small.
The time function of the relative angles between the links can be written as $\phi_i(t)=\varepsilon s_i(t)$ , where $s_i$  represents the ``unscaled'' shape trajectory. Let us also define the vector $\vecs=(s_1,s_2)^T$ .
For the cases of square and circular gaits, the joint angles can be written in the unscaled form as:
\begin{equation}
square: \quad s_1(t)=\left\lbrace
\begin{array}{rl}
1, & t \in [0,2]\\
(3-t), & t \in [2,4]\\
-1, & t \in [4,6]\\
(t-7), & t \in [6,8]
\end{array}\right.,
\quad
s_2(t)=\left\lbrace
\begin{array}{rl}
(1-t), & t \in [0,2]\\
-1, & t \in [2,4]\\
(t-5), & t \in [4,6]\\
1, & t \in [6,8]
\end{array}\right.
\label{eq:squaregait}
\end{equation}
\begin{equation}
circular: \quad s_1(t)=\sin\left(t+\tfrac{\pi}{4}\right), \quad s_2(t)=\cos\left(t+\tfrac{\pi}{4}\right), \qquad t \in [0,2\pi]
\label{eq:circlegait}
\end{equation}
These equations describe circular and square-shaped trajectories with stroke amplitude of 1, which are then scaled to stroke of $\varepsilon$ by setting $\phi_i(t)=\varepsilon s_i(t)$.
Since the equation of motion \eqref{eq:dynamics} is time invariant, the net motion is independent of time parametrization of the gait. Here, time parametrization of the gaits was chosen arbitrarily so that the period times are $T=8$ for the square gait and $T=2\pi$  for the circular one. 
The net displacement in the $x$ direction of one full period will be denoted $X$ and is calculated through $X=\int_0^T\dot{x}dt$. The mean swimming speed is denoted by $V=X/T$.

\subsection{Mechanical power and Efficiency}

The energetic efficiency of a stroke can be defined in a several different ways, cf. \cite{childress.JFM2012}. First, we write an expression for the power exerted by the swimmer.
The mechanical power dissipated by the fluid's viscous drag forces and torques on all three links is,
\begin{equation} \label{eq:power} P=-\sum_{i=0}^2\vecF_i^T\vecV_i=\sum_{i=0}^2\vecV_i^T\mathcal{R}_i\vecV_i.
\end{equation}
On the other hand, the mechanical power expended internally by the actuation torques is
\begin{equation}
	P=\vectau^T\vecphid=\vecphid^T\vecW\vecphid. \label{eq:power2}
\end{equation}
These last two expressions are equivalent, which can be proven using the relations \eqref{eq:W(phi)} and \eqref{eq:tau_vs_phid}. The total work can be calculated from here by $W=\int_0^TP(t)dt$.

Due to time-invariance of the swimmer's dynamics, a well-known observation (cf. \cite{becker2003self,tam2007optimal}) states that the energy expenditure under a given gait trajectory can be made arbitrarily small by pacing along the trajectory in a sufficiently slow rate. Thus, energy per unit distance cannot serve as a reasonable performance measure. Following \cite{becker2003self,tam2007optimal},  we use an energetic efficiency criterion similar to that defined by Lighthill in \cite{lighthill1975mathematica} which is the ratio of average power $\bar{P}=\frac{W}{T}$ exerted by the swimmer to the power needed to drag the swimmer as a rigid body at the same mean speed:
\begin{equation}
	\tilde{\xi}=\dfrac{c_t l V^2}{\bar{P}}=\dfrac{c_t l X^2}{\bar{P}T^2}.
	\label{eq:eff_average}
\end{equation}
This definition is non-unique for a given gait trajectory $\vecphi(\sigma(t))$ and by varying the gait's time parametrization $\sigma(t)$ the average power changes as well. Nevertheless, a known result from \cite{becker2003self} proves that minimal average power $\bar{P}$ for a given gait is obtained by choosing a time parametrization $\sigma(t)$ for which the instantaneous power $P(t)$ is kept constant. Using this particular time parametrization (which is unique up to multiplying by a positive factor) maximizes the energetic efficiency $\xi$ for a given gait, and this efficiency is uniquely determined for any gait trajectory.

The period time $T$ under a constant mechanical power $P=P_o$ can be found using the following derivation. Consider the gait $\vecphi(\sigma)=\vecphi(\sigma(t))$ where $\sigma \in [\sigma_0,\sigma_1]$ is a geometric parameter along the gait's trajectory and $\sigma(t)$ is its time parametrization. Using \eqref{eq:power2}, the instantaneous power $P(t)$ can be written as
\begin{equation}
	P(\sigma(t))=\left (\vecphi'(\sigma)^T \vecW(\sigma)\vecphi'(\sigma)\right )\dot{\sigma}^2 \quad\text{where,}\quad \vecphi'(\sigma)=\dfrac{\partial\vecphi}{\partial \sigma}
\end{equation}
denoting $F(\sigma)=\vecphi'(\sigma)\vecW(\sigma)\vecphi'(\sigma)$ we have
\begin{equation}
	\dot{\sigma}=\dfrac{d \sigma}{dt}=\sqrt{\dfrac{P}{F(\sigma)}} \to dt=\sqrt{\dfrac{F(\sigma)}{P}}d \sigma. \label{eq:sigmaP}
\end{equation}
This transformation is well-defined for any non-degenerate trajectory and time parametrizations such that $\vecphi'(\sigma)\neq 0$ and $d\sigma/dt>0$ for all $\sigma$ and $t$.
The period time of a gait under constant mechanical power $P=P_o$ is thus obtained as
\begin{equation}
	T_{p_o}=\frac{1}{\sqrt{P_o}}\int_{\sigma_0}^{\sigma_1}\sqrt{F(\sigma)}d \sigma
	\label{eq:Time_const_power}
\end{equation}
Substituting \eqref{eq:Time_const_power} into the expression for Lighthill's efficiency (\ref{eq:eff_average}) under a constant power $\bar P =P_o$, it is clear that the $P_o$ cancels out and so the calculations can be done for $P_o=1$ without loss of generality.
Thus, Lighthill's efficiency of a gait $\vecphi(\sigma)$ is uniquely given as:
\begin{equation}
	\xi=\dfrac{X^2}{T_p^2},
	\label{eq:eff}
\end{equation}
where $T_p$ now denotes the period time under constant power of $P_o=1$, as given in \eqref{eq:Time_const_power}, and the constants $c_t l$ are dropped for convenience.

\section{Perturbation expansion of the dynamics}

In order to find the leading-order expressions for the swimmer's motion using perturbation expansion \cite{nayfeh2008perturbation}, the dynamics of the swimmer are expanded as power series of the stroke amplitude $\varepsilon$. First, expansion of the swimmer's orientation angle $\theta(t)$ is obtained, followed by expansions of the instantaneous body velocity and of the resulting net displacement $X$ under each specific gait.


\subsection{Expansion of swimmer displacement}

The position and orientation of the swimmer are now expanded into a power series in $\vep$ as:
\begin{equation}
\vecq(t)=\varepsilon\vecq^{(1)}(t)+\varepsilon^2\vecq^{(2)}(t)+\varepsilon^3\vecq^{(3)}(t)+\ldots
\end{equation}
And for each coordinate:
\begin{equation}
x(t)=\varepsilon x^{(1)}(t)+\varepsilon^2 x^{(2)}(t)+\varepsilon^3\ x^{(3)}(t)+\ldots
\end{equation}
\begin{equation}
y(t)=\varepsilon y^{(1)}(t)+\varepsilon^2 y^{(2)}(t)+\varepsilon^3\ y^{(3)}(t)+\ldots
\end{equation}
\begin{equation}
\theta(t)=\varepsilon \theta^{(1)}(t)+\varepsilon^2 \theta^{(2)}(t)+\varepsilon^3\ \theta^{(3)}(t)+\ldots
\end{equation}
As mentioned, from the matrix $D(\theta)$  in eq. (\ref{eq:gauge_sym}) it can be seen that $\thetad$  is independent of $\theta$ , and thus, the ODE for $\theta$  can be written as:
\begin{equation}
\thetad=\sum_{j=1}^2 \vecG_{3j}(\vecphi)\phid_j=\varepsilon \sum_{j=1}^2 \vecG_{3j}(\vecphi)\dot s_j
\end{equation}
Using Taylor expansion we get:
\begin{IEEEeqnarray}{rCl}
\thetad &=& \sum_{j=1}^2 \Bigg[\vecG_{3j}(0)+\left(\phi_1\frac{\partial}{\partial\phi_1}\Big|_{(0,0)}+\phi_2\frac{\partial}{\partial\phi_2}\Big|_{(0,0)}\right)\vecG_{3j}\phantom{\Bigg(\Bigg)^2}\Bigg.\nonumber\\
&&\Bigg.+\frac{1}{2!}\left(\phi_1\frac{\partial}{\partial\phi_1}\Big|_{(0,0)}+\phi_2\frac{\partial}{\partial\phi_2}\Big|_{(0,0)}\right)^2\vecG_{3j}+\ldots\Bigg]\phid_j\label{eq:theta}
\end{IEEEeqnarray}
Substituting the expansion for $\theta$  and $\phi_i(t)=\varepsilon s_i(t)$ :
\begin{IEEEeqnarray}{l}
\varepsilon \thetad^{(1)}+\varepsilon^2 \thetad^{(2)}+\varepsilon^3\ \thetad^{(3)}+\ldots =\nonumber\\
 \sum_{j=1}^2 \Bigg[\vecG_{3j}(0)+\varepsilon\left(s_1\frac{\partial}{\partial\phi_1}\Big|_{(0,0)}+s_2\frac{\partial}{\partial\phi_2}\Big|_{(0,0)}\right)\vecG_{3j}\phantom{\Bigg(\Bigg)^2}\Bigg.\nonumber\\
\Bigg.+\varepsilon^2\frac{1}{2!}\left(s_1\frac{\partial}{\partial\phi_1}\Big|_{(0,0)}+s_2\frac{\partial}{\partial\phi_2}\Big|_{(0,0)}\right)^2\vecG_{3j}+\ldots\Bigg]\varepsilon \dot s_j
\end{IEEEeqnarray}
Now we can collect terms of different orders in $\varepsilon$.
\begin{equation}
\thetad^{(1)}(t)=\sum_{j=1}^2\vecG_{3j}(0) \dot s_j(t)
\label{eq:thetadot1}
\end{equation}
\begin{equation}
\thetad^{(2)}(t)=\sum_{j=1}^2 \left( \frac{\partial\vecG_{3j}}{\partial\phi_1}(0)s_1 + \frac{\partial\vecG_{3j}}{\partial\phi_2}(0)s_2 \right) \dot s_j(t)
\label{eq:thetadot2}
\end{equation}
\begin{equation}
\thetad^{(3)}(t)=\sum_{j=1}^2 \left( \frac{1}{2}\frac{\partial^2\vecG_{3j}}{\partial\phi_1^2}(0)s_1^2 + \frac{1}{2}\frac{\partial^2\vecG_{3j}}{\partial\phi_2^2}(0)s_2^2 + \frac{\partial^2\vecG_{3j}}{\partial\phi_1\partial\phi_2}(0)s_1s_2\right) \dot s_j(t)
\label{eq:thetadot3}
\end{equation}
Explicit expressions for the derivatives are given in the supplementary document.
The first-order derivatives at the origin are zero since $\vecG_{31},\vecG_{32}$ are even functions in $(\phi_1,\phi_2)$  and so  $\thetad^{(2)}=0$ in \eqref{eq:thetadot2}, implying that $\theta^{(2)}(t)=0$ is zero (under zero initial conditions).
Now, using the expression for $\theta(t)$  we can calculate expansions of the net motion $\vecq(t)$. First, we expand the matrices in \eqref{eq:gauge_sym} as:
\begin{equation}
\begin{array}{c}
\vecD(\theta)=\vecI+\theta\vecJ+\frac{1}{2!}\theta^2\vecJ^2+\frac{1}{3!}\theta^3\vecJ^3+\ldots\\[12pt]
\vecJ=\mathree{0}{-1}{0}{1}{0}{0}{0}{0}{0}
\end{array}
\end{equation}
\begin{equation}
\begin{array}{rcl}
\vecG(\vecphi) &=& \vecG(0)+\varepsilon\left(s_1\dfrac{\partial}{\partial\phi_1}\Big|_{(0,0)}+s_2\dfrac{\partial}{\partial\phi_2}\Big|_{(0,0)}\right)\vecG \\[12pt]
&& +\varepsilon^2\frac{1}{2!}\left(s_1\dfrac{\partial}{\partial\phi_1}\Big|_{(0,0)}+s_2\dfrac{\partial}{\partial\phi_2}\Big|_{(0,0)}\right)^2\vecG+\ldots
\end{array}
\end{equation}
This, of course, is the same process done for the last row of $\vecG$  
in (\ref{eq:theta}).
Expanding equation (\ref{eq:gauge_sym}) by substituting expansions for $\theta(t)$, $\vecD(\theta)$, $\vecG(\vecphi)$  and also of $\vecq$ and $\phi_i(t)=\varepsilon s_i(t)$ and collecting orders of $\varepsilon$, one obtains:
\begin{equation}
\resizebox{\textwidth}{!}{$\begin{array}{lccl}
\varepsilon^1: & \vecqd^{(1)} &=& \vecG(0)\dot{\vecs}(t)\\[12pt]
\varepsilon^2: & \vecqd^{(2)} &=& \left[\theta^{(1)}\vecJ\vecG(0)+ \left(s_1\dfrac{\partial}{\partial\phi_1}\Big|_{(0,0)}+s_2\dfrac{\partial}{\partial\phi_2}\Big|_{(0,0)}\right)\vecG\right]\dot{\vecs}(t)\\[12pt]
\varepsilon^3: & \vecqd^{(3)} &=& \Bigg[
\frac{1}{2}\theta^{(1)2}\vecJ^2\vecG(0)+ \theta^{(1)}\vecJ\left(s_1\dfrac{\partial}{\partial\phi_1}\Big|_{(0,0)}+s_2\dfrac{\partial}{\partial\phi_2}\Big|_{(0,0)}\right)\vecG\phantom{\Bigg(\Bigg)^2}\Bigg.\\[10pt]
&&&+\Bigg.\frac{1}{2}\left(s_1\dfrac{\partial}{\partial\phi_1}\Big|_{(0,0)}+s_2\dfrac{\partial}{\partial\phi_2}\Big|_{(0,0)}\right)^2\vecG
\Bigg]\dot{\vecs}(t)\\[12pt]
\varepsilon^4: & \vecqd^{(4)} &=& \Bigg[ \left(\theta^{(3)}\vecJ+\frac{1}{3!}\theta^{(1)3}\vecJ^3\right)\vecG(0)+
\frac{1}{2}\theta^{(1)2}\vecJ^2\left(s_1\dfrac{\partial}{\partial\phi_1}\Big|_{(0,0)}+s_2\dfrac{\partial}{\partial\phi_2}\Big|_{(0,0)}\right)\vecG\Bigg.\\[10pt]
&&&+ \frac{1}{2}\theta^{(1)}\vecJ\left(s_1\dfrac{\partial}{\partial\phi_1}\Big|_{(0,0)}+s_2\dfrac{\partial}{\partial\phi_2}\Big|_{(0,0)}\right)^2+\Bigg.\frac{1}{6}\left(s_1\dfrac{\partial}{\partial\phi_1}\Big|_{(0,0)}+s_2\dfrac{\partial}{\partial\phi_2}\Big|_{(0,0)}\right)^3\vecG
\Bigg]\dot{\vecs}(t)\\[12pt]
\end{array}$}
\label{eq:velocities}
\end{equation}
Explicit expressions for the derivatives of elements of $\vecG$ are given in the supplementary document. Due to symmetry of the swimmer and of the gaits, it is known that there will only be net translation along $x$ direction, which is the axis of the center link, while rotation and translation in $y$ direction are cancelled out \cite{Gutman2015Symmetries}. Thus, expansions of net displacement $X$ for specific gaits are calculated next.

\subsection{Gait-specific expressions of the displacement}
Once the series expansion for the instantaneous velocities of the swimmer is derived, it is possible to obtain the net displacement $X$ over a period for a specific gait via integration. The following proposition summarizes the results for the square and circular gaits:

\begin{prop}\label{th:x_disp}
For a symmetric, three linked "Purcell swimmer" performing a square or circular gait with amplitude $\varepsilon$, the leading-order term and next-order correction for the displacement $X$ over one full stroke in the direction of $x$ axis  are given as:
\begin{equation}
X=f_2(\eta)\varepsilon^2-f_4(\eta)\varepsilon^4+O(\varepsilon^6)
\label{eq:disp_expansion}
\end{equation}
With:
\begin{IEEEeqnarray}{l}
f_2(\eta)=C_2\eta l(1-\eta)^3(\eta+3)\nonumber\\
f_4(\eta)=C_4\eta l(1-\eta)^3(\eta^7+3\eta^6-10\eta^5-22\eta^4+29\eta^3+95\eta^2+44\eta+20)
\label{eq:disp_f2f4}
\end{IEEEeqnarray}
Where, for the square gait given in (\ref{eq:squaregait}) we have $C_2=1/4,C_4=1/192$, and for the circular gait given in (\ref{eq:circlegait}) $C_2=\pi/16,C_4=\pi/1024$.
\end{prop}

\noindent
The proof of this proposition is given in the supplementary document. It can easily be seen that the leading order term of $X$ is quadratic in $\vep$ and that the next-order corrections are of opposite sign to the leading-order terms. This implies that the displacement grows monotonically with $\vep$ for small stroke amplitudes, whereas for larger amplitudes with $\vep>1$ the $O(\vep^4)$-term causes reversal in the direction of net motion. This, in turn, indicates the existence of a locally optimal value of $\vep$ that achieves maximal displacement. Note that, a leading-order term for displacement under the square gait is found in \cite{giraldi2015optimal} using Lie brackets, and the results are identical to those given here.

\begin{figure}[!b]
\vspace{-10pt}
    \centering
    \begin{subfigure}[h]{0.49\textwidth}
    \includegraphics[width=\textwidth]{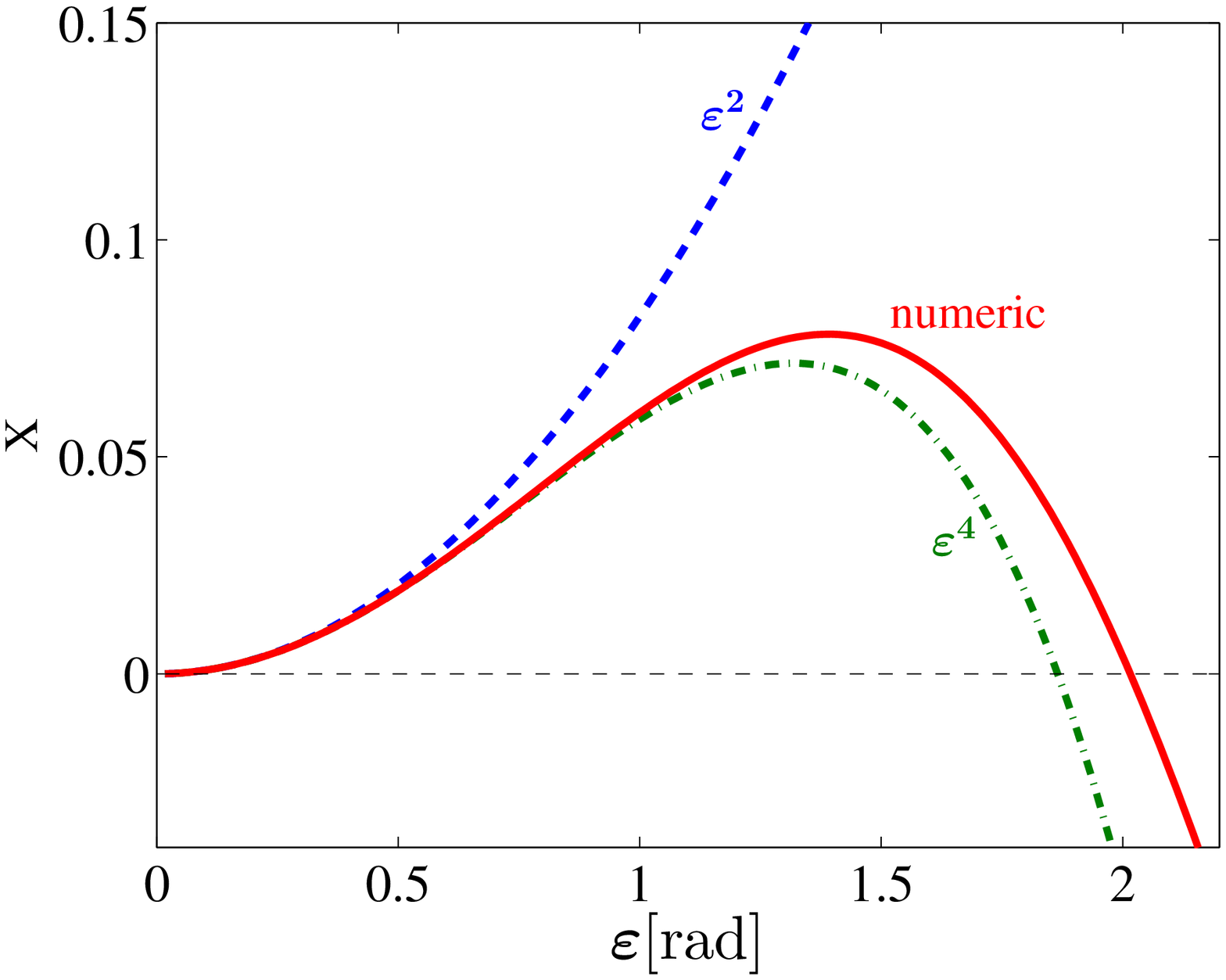}
    \caption{}
    \end{subfigure}
    \begin{subfigure}[h]{0.49\textwidth}
    \includegraphics[width=\textwidth]{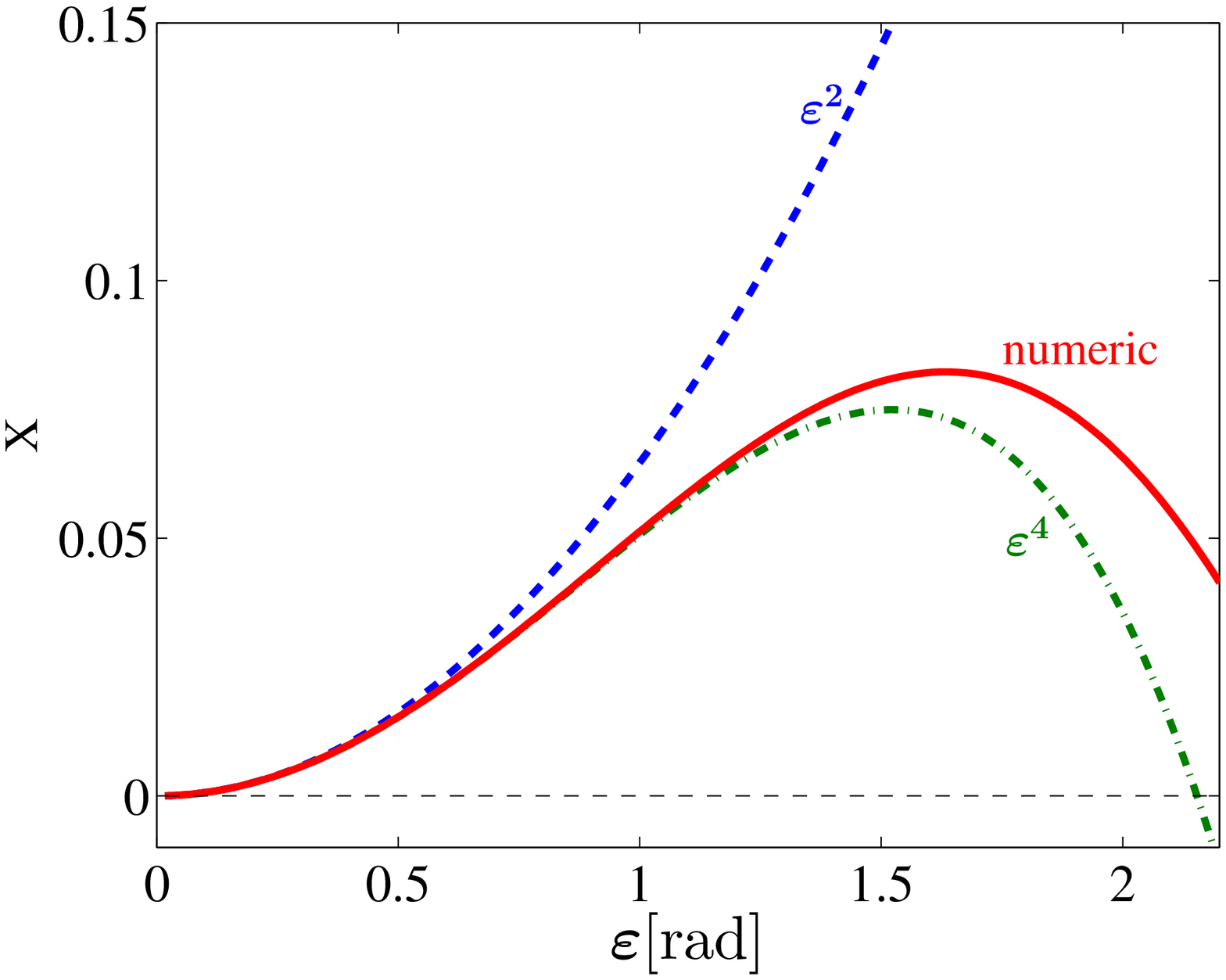}
    \caption{}
    \end{subfigure}
    \caption{Displacement over one cycle $X$ for $\eta=1/3$ as a function of stroke amplitude $\vep$ for (a) Square gait, (b) Circular gait. Solid curves - numerical integration. Dashed curves - $O(\vep^2)$ approximation. Dash-dotted curves - $O(\vep^2)$ approximation.}
    \label{fig:trans_ep}
\end{figure}

As a demonstration of the utility of Proposition \ref{th:x_disp}, Figure \ref{fig:trans_ep} shows plots of the displacement $X$ vs. stroke amplitude $\varepsilon$ with link ratio of $\eta=1/3$ for both gaits (a - square, b - circular). The solid curves are obtained from numerical integration of the nonlinear equations of motion \eqref{eq:dynamics}, whereas the dashed and dash-dotted curves are $O(\vep^2)$ and $O(\vep^4)$ approximations, respectively. While the $O(\vep^2)$-approximation is monotonic in $\vep$ and works only for small stroke amplitudes, the $O(\vep^4)$-approximation with next-order correction captures the reversal in the direction of the displacement for intermediate amplitudes, though there is an increasing deviation from the numerical values for larger amplitudes.

\subsection{Optimal geometry and stroke amplitude for maximal displacement}
\label{sec:Optimal geometry for maximal displacement}
First, we discuss the dependence of net displacement $X$ on the swimmer's geometric ratio $\eta$ and derive its locally optimal values for both gaits, based on equations \eqref{eq:disp_expansion} and \eqref{eq:disp_f2f4}. As a demonstration, plots of $X$ vs. $\eta$ under the square gait are shown in Figures \ref{fig:trans_vs_eta}(a) and \ref{fig:trans_vs_eta}(b) for stroke amplitudes of $\varepsilon=\pi/4$ and $\varepsilon=\pi/2$, respectively. It can be seen that the $O(\vep^2)$ approximation (dashed) has a larger deviation from the exact value obtained by numerical integration (solid curves), compared to that of the  $O(\vep^4)$ approximation (dash-dotted), where the deviations are further increased for larger amplitudes $\vep$. Nevertheless, in all cases it is obvious that the displacement $X$ vanishes at extreme cases of $\eta \to 0$ or $\eta \to 1$, where either the middle link or side links vanish, and that an optimal value of $\eta$ exists that maximizes the displacement $X$.

\begin{figure}[!t]
    \centering
    \begin{subfigure}{0.49\textwidth}
    \includegraphics[width=\textwidth]{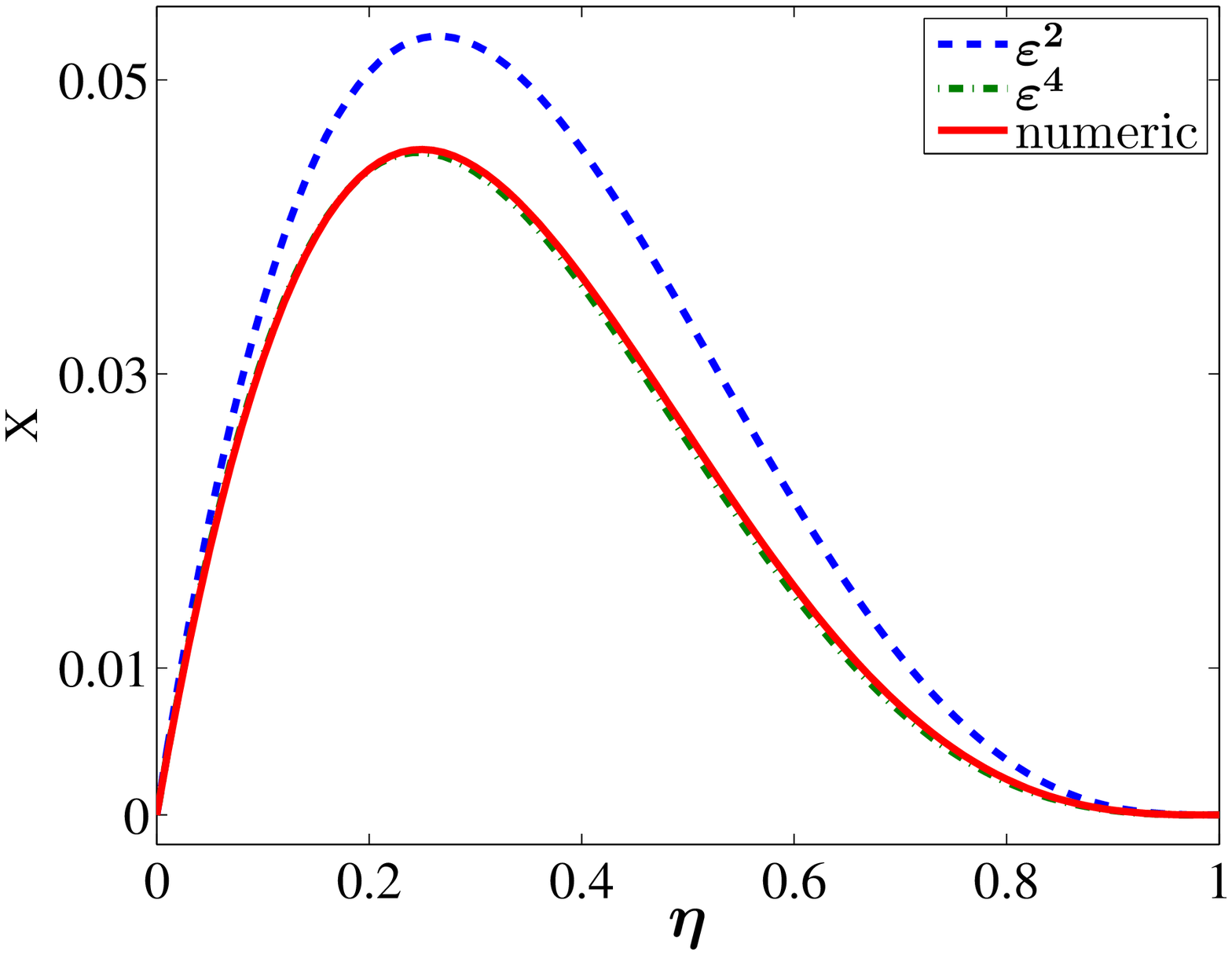}
    \caption{}
    \label{fig:trans_vs_eta_eps_qrtrpi}
    \end{subfigure}
    \begin{subfigure}{0.49\textwidth}
    \includegraphics[width=\textwidth]{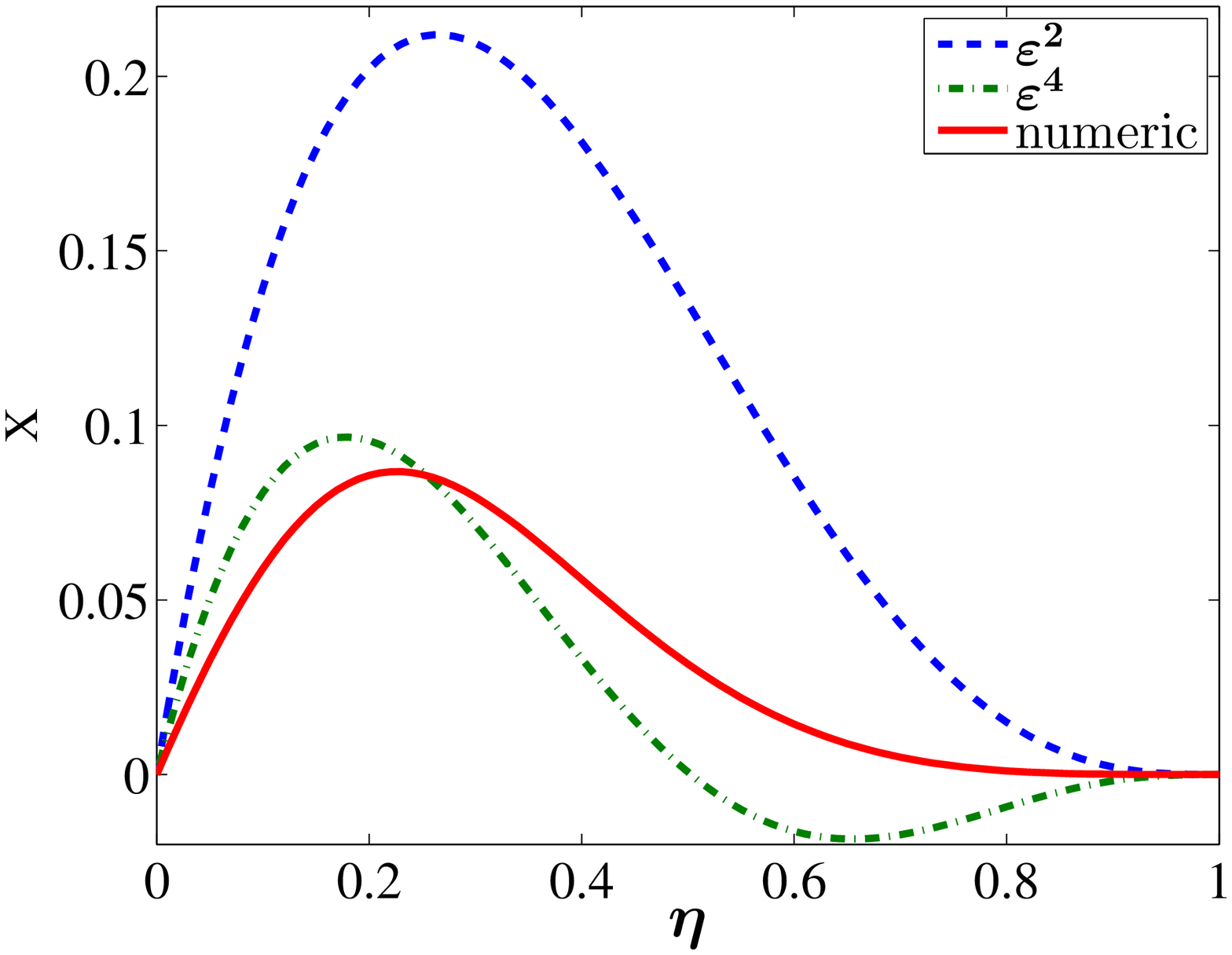}
    \caption{}
    \label{fig:trans_vs_eta_eps_halfpi}
    \end{subfigure}
    \caption{Displacement over one cycle $X$ as a function of $\eta$ under the square gait with stroke amplitudes of (a) $\varepsilon=\pi/4$, (b) $\varepsilon=\pi/2$. Solid curves - numerical integration. Dashed curves - $O(\vep^2)$ approximation. Dash-dotted curves - $O(\vep^4)$ approximation.}
    \label{fig:trans_vs_eta}
\vspace{-10pt}
\end{figure}

We now obtain approximations of the optimal value of $\eta$ 
based on the leading-order $O(\vep^2)$ terms in \eqref{eq:disp_expansion}. Interestingly, it can be seen from \eqref{eq:disp_f2f4} that the leading-order expressions for both square and circular gaits are identical up to multiplication by a constant. An important observation is that this relation is fairly general. That is, for any small-amplitude gait and any swimmer's dynamics with the same structure (i.e. not necessarily assuming resistive force theory), the leading-order term of $X$ can always be decomposed into a function of swimmer's geometry multiplied by a function of the gait's unscaled trajectory.  This key statement is summarized in the following theorem:
\begin{theorem}\label{th:global_eta}
Consider a swimmer with two shape variables whose dynamics can be written in the form of equation \eqref{eq:gauge_sym}, under a small-amplitude gait $\vecphi(t)=\vep \vecs(t)$. The leading-order approximation of the swimmer's displacement $X$ over one period can be written in the form: $X^{(2)}=C(\vecs)   \cdot f(\eta)$, where $C(\vecs)$ depends on the shape of the unscaled trajectory and  $f(\eta)$ depends on the swimmer's geometric structure.
\end{theorem}
The theorem, whose proof is detailed in the supplementary document, implies that based on leading order approximation, the optimal geometry is independent of the gait's shape.
For our particular swimmer model, the geometry-dependent function in \eqref{eq:disp_expansion} is $\eta l(1-\eta)^3(\eta+3)$. As expected, for the cases of $\eta=0$ and $\eta=1$ the net displacement vanishes, and the the optimal value of $\eta$  is easily obtained as $\eta^*=0.4\sqrt{10}-1=0.2649$, which results in maximal displacement of $X^*=0.0859\varepsilon^2$ for the square gait and $X^*=0.0675\varepsilon^2$ for the circular gait. The value of $\eta^*$ is in agreement with the optimal geometry obtained in \cite{giraldi2015optimal} but not with the one in \cite{becker2003self}. This disagreement, which originates from differences in definitions and scaling, is discussed in the sequel.

Next, we use both $O(\vep^2)$ and $O(\vep^4)$ terms in (\ref{eq:disp_expansion})  in order to obtain both optimal ratio $\eta$  and amplitude $\varepsilon$ for the two gaits.
The polynomials $f_2(\eta)$ and $f_4(\eta)$ in \eqref{eq:disp_f2f4} are the same for both gaits up to multiplication by a scalar. However, this is not true for all gaits. For example, it can be shown that for an elliptical gait trajectory, $f_4(\eta)$ is a different polynomial. For a given value of $\eta$, \eqref{eq:disp_expansion} implies that the locally optimal amplitude $\varepsilon^*$  for maximal displacement is given by
\begin{equation}
\varepsilon^*=\sqrt{\dfrac{f_2(\eta)}{2f_4(\eta)}}.
\label{eq:optamp}
\end{equation}
Substituting into (\ref{eq:disp_expansion}) one obtains the optimal displacement as
\begin{equation}
X^*=\dfrac{f_2(\eta)}{4f_4(\eta)}
\label{eq:optdisp}
\end{equation}
This implies the existence of an optimal combination of link ratio $\eta$ and amplitude $\varepsilon$ that achieve a local maximum of the displacement. Differentiating \eqref{eq:optdisp} with respect to $\eta$, the optimal value is obtained by finding the roots of a 9th order polynomial in $\eta$.
Since the expression for $X^*$ in \eqref{eq:optdisp} is  the same for both square and circular gaits up to a multiplicative constant, the optimal geometry for both gaits is obtained as $\eta^*=0.1784$.
Substituting into (\ref{eq:optamp}), optimal amplitudes $\vep$ and displacements $X$ are obtained for each gait and are presented in Table \ref{table:disp_optimal_values}.
Numerically calculating the optimal geometry $\eta$ and stroke amplitude $\vep$ using MATLAB's function {\bf fmincon} gives the results also presented in Table \ref{table:disp_optimal_values}.

The locally optimal amplitude $\varepsilon^*$ based on $O(\vep^4)$ approximation is quite close to the exact value obtained numerically. On the other hand, the large difference in $\eta^*$ can be explained by the large deviation in $X$ for $\vep \approx \pi /2$, as seen in Fig. \ref{fig:trans_vs_eta}(b). The figure also shows that the $O(\vep^4)$ expression for $X$ overestimates the exact value. Finally, it is important to note that all optimization results discussed here are limited to finding local rather than global maxima. In fact, continuing to higher order terms in the expansions of \eqref{eq:disp_expansion} reveals a minimum of $X<0$ for $\vep \approx 3 \; rad$ with larger absolute value than the first positive maximum, and even higher extremum points for larger nonphysical values of $\vep$, see \cite{huber2011micro}.
\begin{table}[t]
\centering
\begin{tabular}{|>{$}c<{$}|>{$}c<{$}|>{$}c<{$}|}
\hline
	     & \text{Analytic}				   & \text{Numeric}
	     \\\hline
                & \varepsilon^* =1.568 \text{rad}  & \varepsilon^*=1.434 \text{rad}
\\
\text{Square gait}   & X^*=0.0967					   & X^*=0.0905
\\
	            & \eta^* = 0.1784				   & \eta^* = 0.2284
\\\hline
                & \varepsilon^* =1.810 \text{rad}  & \varepsilon^*=1.702 \text{rad}
\\
\text{Circular gait} & X^*=0.1012					   & X^*=0.0973
\\
		        & \eta^* = 0.1784				   & \eta^* = 0.2218
\\\hline
\end{tabular}
\caption{Values for maximal displacement}
\label{table:disp_optimal_values}
\vspace{-10pt}
\end{table}

\subsection{Comparison to optimal geometry obtained by Becker {\em et al}}
As mentioned above, the work of Becker {\em et al} \cite{becker2003self} obtained a different value of optimal geometry $\eta$. The explanation for this difference is twofold: first, \cite{becker2003self} considered maximization of the mean forward velocity under a constant torque at the joint rather than net displacement. Second, \cite{becker2003self} scaled swimming distance by the side link's length $l_1$ rather than the total length $l$.

We now show that by adopting the definitions of \cite{becker2003self} and using our leading-order expressions in Proposition \ref{th:x_disp}, one obtains optimal geometry that agrees with \cite{becker2003self}. First, we consider the case where a constant torque is applied on the rotating joint in each part of the square gait. The relation between the joint torques and the joint velocities is given in (\ref{eq:torque}). For the first quarter of the square gait  we have that $\phid_1=0$ and so the second row of (\ref{eq:torque}) is reduced to $\tau_2=W_{22}\phid_2$, with $W_{ij}$ the elements of $\vecW(\vecphi)$ defined in (\ref{eq:W(phi)}).
Assuming a constant torque of $\tau_2(t)=\tau_o$ at the joint for the first quarter cycle of the square gait, we can derive a leading-order expression for the time it would take to complete the quarter gait and then multiply by four in order to obtain:
\begin{IEEEeqnarray}{lcl}
T_{\tau} &=&  \int dt = 4 \int_\varepsilon^{-\varepsilon} \dfrac{dt}{d \phi_2} d \phi_2 = 4 \int_\varepsilon^{-\varepsilon}\dfrac{W_{22}(\vecphi)}{\tau_o}d\phi_2 \nonumber \\ & = &  4 \int_\varepsilon^{-\varepsilon}\dfrac{W_{22}(0)}{\tau_o}d\phi_2
+O(\varepsilon^3)=\dfrac{1}{12}l^3(1\!-\!\eta^2)^3\dfrac{\ct}{\tau_o}
\varepsilon+O(\varepsilon^3)\nonumber\\
\label{eq:Period_time_const_torque} \end{IEEEeqnarray}
The mean forward velocity $V_\tau=X/T_\tau$ can now be computed to leading order as
\begin{equation} \label{eq:V1}
V_\tau^{(1)} =\dfrac{X^{(2)}}{T_\tau^{(1)}}=\varepsilon \dfrac{3\eta(\eta+3)}{l^2(\eta+1)^3}\dfrac{\tau_o}{\ct}.
\end{equation}
From \eqref{eq:V1}, the optimal geometry that maximizes the leading-order mean velocity $V^{(1)}$ is obtained as $\eta^*=0.646$, which is fundamentally different from the optimal geometry for maximal displacement $X$ obtained above.

Becker {\em et al} \cite{becker2003self} also used a different definition of geometric ratio, denoted here as $\eta_b=l_0/l_1$. That is, the length is normalized by that of the side links. The relation between the two definitions of $\eta$ is given by
\begin{equation}
\eta_b=\dfrac{l_0}{l_1}=\dfrac{2\eta}{1-\eta}.
\label{eq:etab}
\end{equation}
Leading-order expressions for the time of quarter period under constant torque $T_\tau$ and the mean velocity $V$ are given in \cite{becker2003self} as (after multiplying $T_{\tau}$ by four and adapting some notation):
\begin{IEEEeqnarray}{l}
T_{\tau}=\left[\dfrac{16}{3}\dfrac{(\eta_b+1)^3l_1^3}{(\eta_b+4)(\eta_b+1)^2+3\eta_b+4}+O(\varepsilon^2)\right]\left(\dfrac{\ct l_0^3}{\tau_o}\right),\\
V_\tau=\left[\dfrac{3}{4}\dfrac{\eta_b(2\eta_b+3)}{(\eta_b+1)^3(\eta_b+2)l_1^2}+O(\varepsilon^2)\right]\left(\dfrac{\tau_o}{\ct l_0^2}\right). \label{eq:Vbecker}
\end{IEEEeqnarray}
Substituting the transformation (\ref{eq:etab}), it can be shown that these expressions agree with our derivations in (\ref{eq:Period_time_const_torque}),\eqref{eq:V1}, and also with our leading order expression for net displacement $X=V_\tau T_\tau$ as given in \eqref{eq:disp_expansion}.  From \eqref{eq:Vbecker},  optimal geometry for maximizing $V_\tau$ has been obtained in \cite{becker2003self} as $\eta_b^*=0.54$, which corresponds to $\eta=0.213$. The difference from the optimal value of $\eta=0.646$ obtained by maximizing $V_\tau$ in \eqref{eq:V1} is now due to the fact that \eqref{eq:V1} is scaled by $l$ while \eqref{eq:Vbecker} is scaled by $l_1$.

\subsection{Period time under constant power }
We now compute series expansions for the period time $T_p$ under constant power $P_0=1$ as formulated in \eqref{eq:Time_const_power}, for both square and circular gaits. These expressions are used in the next section for obtaining approximate expressions of Lighthill's energetic efficiency according to equation \eqref{eq:eff}.
The period time is expanded into a power series as
\begin{equation}
 	T_p=\varepsilon T_{p}^{(1)}+\varepsilon^2 T_{p}^{(2)}+\varepsilon^3 T_{p}^{(3)}+\ldots .
 \end{equation}
Explicit expressions for the different terms $T_p^{(i)}$ depend on the choice of the gait. For the square and circular gaits, these expressions are summarized in the following proposition.


\begin{prop}
For a symmetric, three linked "Purcell swimmer" performing a square or circular gait with amplitude $\varepsilon$, the period times of one full stroke with constant power of $P_o=1$ exerted by the joints are expanded as follows.
\newline
Square gait:

\begin{IEEEeqnarray}{rcl}
T_{p} &=&\frac{\sqrt{6}}{3}\sqrt{\ct l^3 {\left({1 - \eta}^2\right)}^3}\:\varepsilon \nonumber\\&&+ \frac{\sqrt{6  \ct l^3}  {\left( 1-\eta \right)}^4 \left(\eta+1\right) \left(\eta^3 + 3\eta^2 - 3\eta + 1\right)}{12  \sqrt{{\left(1 - \eta^2\right)}^3}}\varepsilon^3+O(\varepsilon^5)
\label{eq:sq_Time_constpower}
\end{IEEEeqnarray}
Circular gait:
\begin{IEEEeqnarray}{l}
T_p=E\left(\frac{\eta^3+3\eta^2-3\eta-1}{\eta^2(\eta+3)}\right)\sqrt{\frac{c_tl^3}{3}\eta^2(1-\eta)^3(\eta+3)}\:\varepsilon+O(\varepsilon^3)
\label{eq:cir_Time_constpower}
\end{IEEEeqnarray}
\end{prop}

The function $E(\cdot)$ in \eqref{eq:cir_Time_constpower} denotes a complete elliptic integral of the second kind, whose definition is reviewed in the supplementary document.
The proof of this proposition is also given in the supplementary document.

In order to validate the expansions of $T_p$ in \eqref{eq:sq_Time_constpower} and \eqref{eq:cir_Time_constpower}, Figure \ref{fig:period_time} plots the $O(\vep)$ and $O(\vep^3)$ approximations of $T_p$ as a function of stroke amplitude $\vep$, compared to the exact value of $T_p$ obtained by numerical integration of \eqref{eq:Time_const_power}, for both gaits (a - square, b - circular) under equal links length, $\eta = 1/3$. It can be seen that for small amplitudes, the period time $T_p$ indeed scales linearly with $\vep$ and that for intermediate values of $\vep$ the $O(\vep^3)$ correction term slightly improves the approximation accuracy (square gait only).

\begin{figure}[!t]
    \centering
    \begin{subfigure}[h]{0.49\textwidth}
    \includegraphics[width=\textwidth]{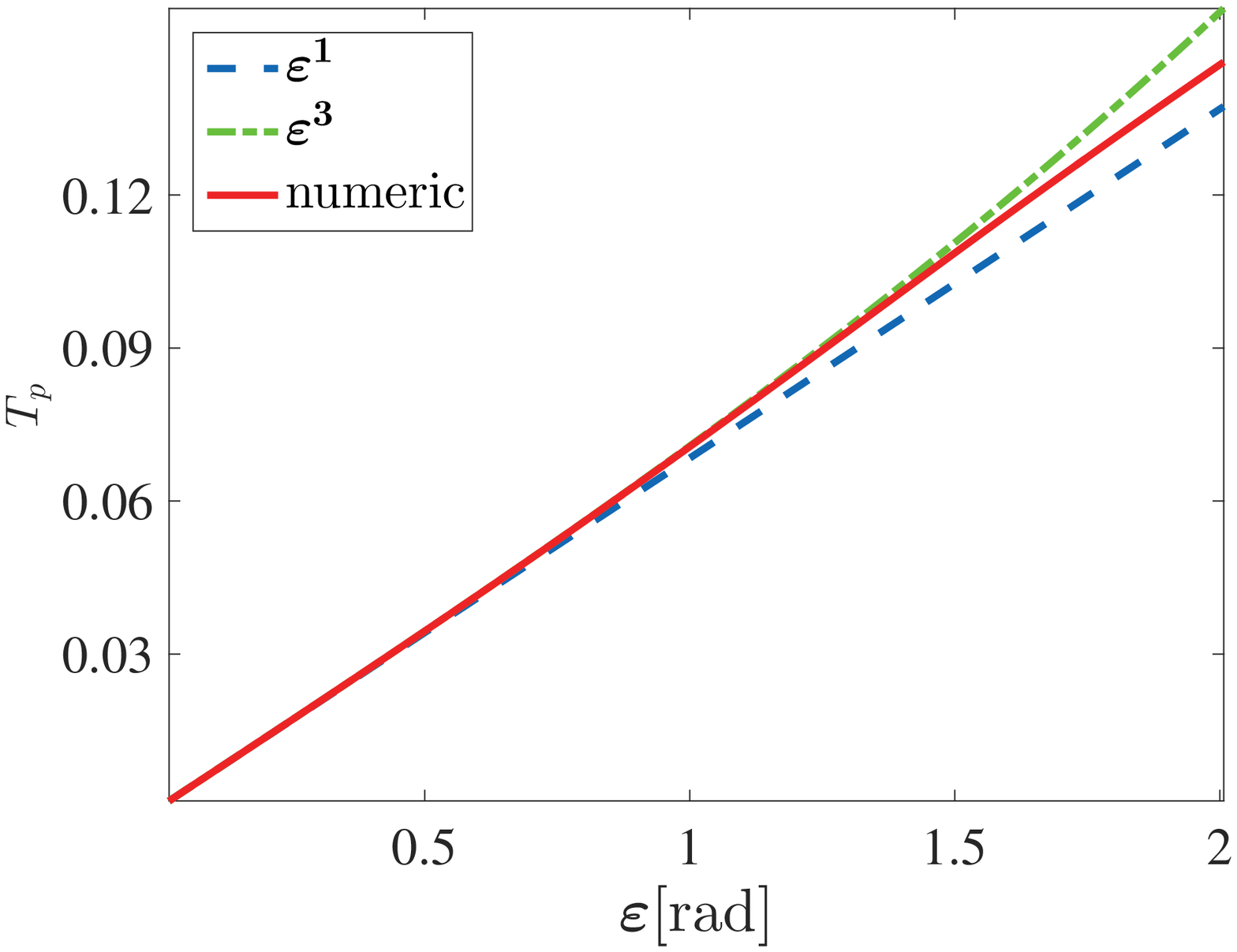}
    \caption{}
    \end{subfigure}
    \begin{subfigure}[h]{0.49\textwidth}
    \includegraphics[width=\textwidth]{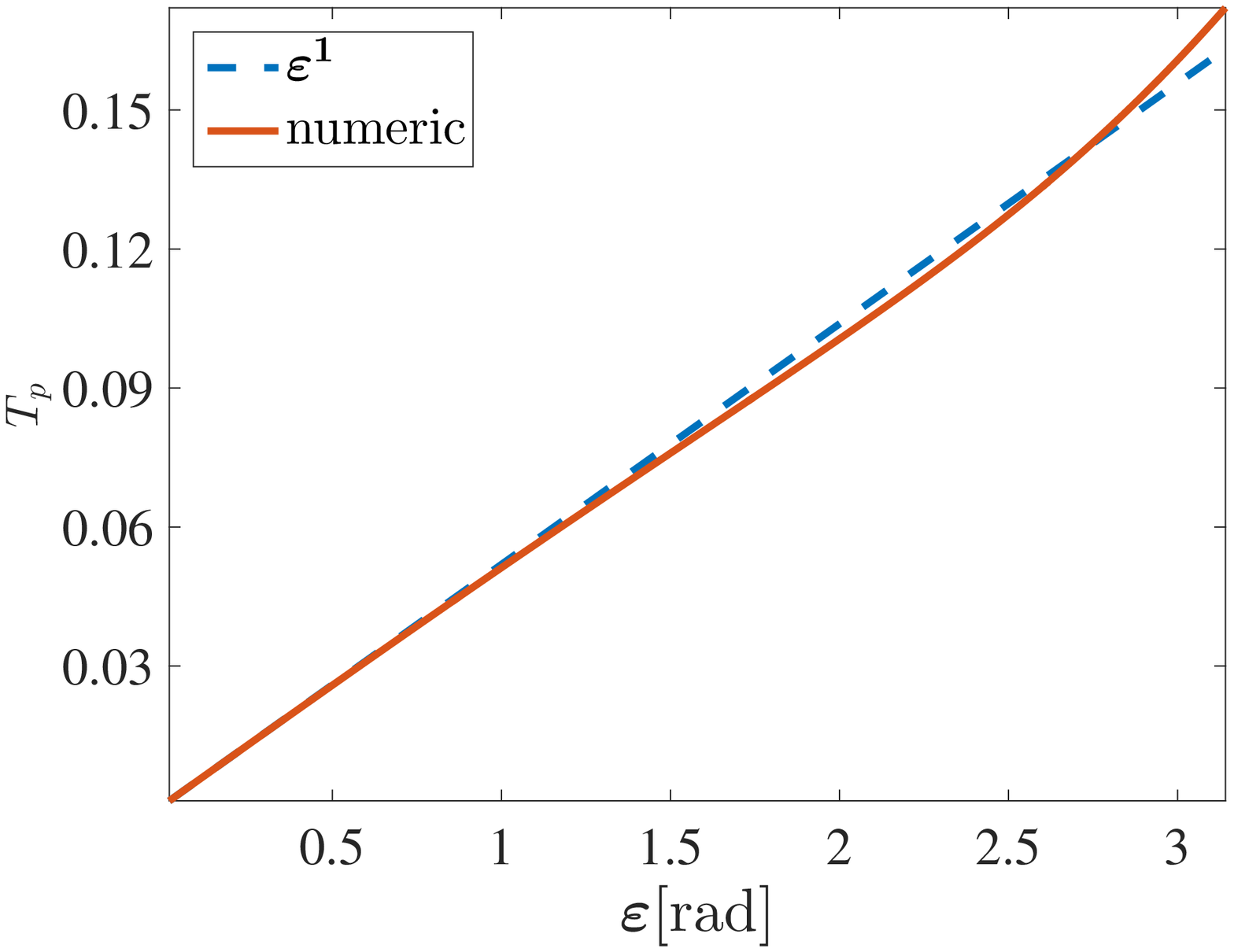}
    \caption{}
    \end{subfigure}
    \caption{Period time under constant power as a function of stroke amplitude $\vep$ for $\eta=1/3$, comparison between exact and approximate expressions. (a) square gait, (b) circular gait. }  
    \label{fig:period_time}%
    \vspace{-10pt}
\end{figure}

\section{Approximate expressions for Lighthill's efficiency}

Using the expressions found in (\ref{eq:sq_Time_constpower}),(\ref{eq:cir_Time_constpower}) for the period times under constant power, along with previous results of net displacement (\ref{eq:disp_expansion}), we now calculate several approximations of Lighthill's energetic efficiency $\xi$ using \eqref{eq:eff}. Then we obtain optimal values of length ratio $\eta$ and stroke amplitude $\vep$ for maximal energetic efficiency, under both square and circular gaits. Using the expansions for $X$ and $T_p$, the efficiency $\xi$ can be written as
\begin{equation}
	\xi=\dfrac{X^2}{T_p^2}=\vep^2 \dfrac{\left( X^{(2)}+ \vep^2 X^{(4)}+ \cdots \right)^2}{\left( T_p^{(1)}+ \vep^2 T_p^{(3)}+ \cdots \right)^2}.
	\label{eq:eff_ep}
\end{equation}
Calculating Taylor expansion of \eqref{eq:eff_ep} in powers of $\vep$, one obtains:
\begin{equation}
	\xi=\vep^2 \dfrac{\left(X^{(2)}\right)^2}{\left(T_p^{(1)}\right)^2}+2\vep^4
X^{(2)}\dfrac{T_p^{(1)}X^{(4)}-X^{(2)}T_p^{(3)}}{\left(T_p^{(1)}\right)^3}
+O(\vep^6).
	\label{eq:eff_expansion}
\end{equation}
Explicit approximations for the efficiency $\xi$ under square and circular gaits, are discussed below.

\subsection{Square gait}
Substituting \eqref{eq:disp_expansion},\eqref{eq:disp_f2f4} and \eqref{eq:sq_Time_constpower} into \eqref{eq:eff_expansion}, the first two elements in the expansion of Lighthill's energetic efficiency under the square gait are given by:
\begin{equation}
	\xi^{(2)}=\dfrac{\left(X^{(2)}\right)^2}{\left(T_p^{(1)}\right)^2}=\dfrac{3 \, \eta^2 (1-\eta)^3(\eta+3)^2}{32 \,\ct\, l {\left(\eta+1\right)}^3} \label{eq:xi2sq}
\end{equation}
\begin{equation}
\xi^{(4)}=\frac{\eta^2 \varepsilon^4 {\left(\eta \! -\! 1\right)}^3 \left(\eta \! + \! 3\right) \left(\eta^9 \! + \! 5 \eta^8 \! - \! 3 \eta^7 \! - \! 39\, \eta^6 \! - \! 37 \eta^5 \! + \! 71 \eta^4 \! + \! 263 \eta^3 \! + \! 371 \eta^2 \! - \! 48 \eta \! + \! 56\right)}{256 c_t l {\left(\eta \! + \! 1\right)}^5}
\end{equation}
Figure \ref{fig:xi_vs_eps_sq} plots the $O(\vep^2)$ and $O(\vep^4)$ approximations of $\xi$ as a function of $\vep$ in dashed and dash-dotted curves,  respectively, under the square gait with length ratio of $\eta=1/3$. For comparison, the exact value of $\xi$ obtained from numerical integration is plotted in solid curve. While the $O(\vep^2)$ approximation grows monotonically with $\vep$, adding the next-order correction of $O(\vep^4)$ also captures a local maximum in the efficiency, which is obtained at an intermediate stroke amplitude of $\vep \approx 1$. Nevertheless,  this approximation of $\xi$ becomes negative for larger values of $\vep$, which is nonphysical according to the definition of $\xi$ in \eqref{eq:eff}. A more reasonable approximation is obtained by plugging the $O(\vep^4)$ and $O(\vep^3)$ approximations for $X$ and $T_p$, respectively, into \eqref{eq:eff}. With a slight abuse of notation, this approximation is denoted here as $O(\vep^{4/3})$, and is given by \begin{equation}
\xi^{(4/3)}=\dfrac{\left(\varepsilon^2X^{(2)}+\varepsilon^4X^{(4)}\right)^2}{\left(\varepsilon T_p^{(1)}+\varepsilon^3T_p^{(3)}\right)^2}. \label{eq:xi43}
\end{equation}
Note that this approximation is different from the direct expansion of \ref{eq:eff_ep} into \ref{eq:eff_expansion}, as demonstrated below.
A plot of this approximation is overlayed as a dotted line in Figure \ref{fig:xi_vs_eps_sq}. It can be seen that the $O(\vep^{4/3})$ approximation is always positive, and correctly captures the general trend of $\xi$ as a function of $\varepsilon$. Finally, it can be seen from the plot that the exact value of $\xi$ attains a global maximum value at $\vep \approx 3$, where the direction of swimming is reversed. This second maximum is not captured by any of the approximate expressions mentioned above. Nevertheless, as already mentioned in \cite{becker2003self}, this maximum  at large strokes where the joint angles approach $\pm \pi$ is often impractical, due to possible collisions between the links. Moreover, resistive force theory which assumes negligible hydrodynamic interaction between the links is no longer valid in this range \cite{huber2011micro}. Thus, we focus here on approximations of the first maximum of $\xi$ which is attained at moderate amplitudes, as discussed next.

\begin{figure}[!b]
    \centering
    \begin{subfigure}[h]{0.49\textwidth}
    \includegraphics[width=\textwidth]{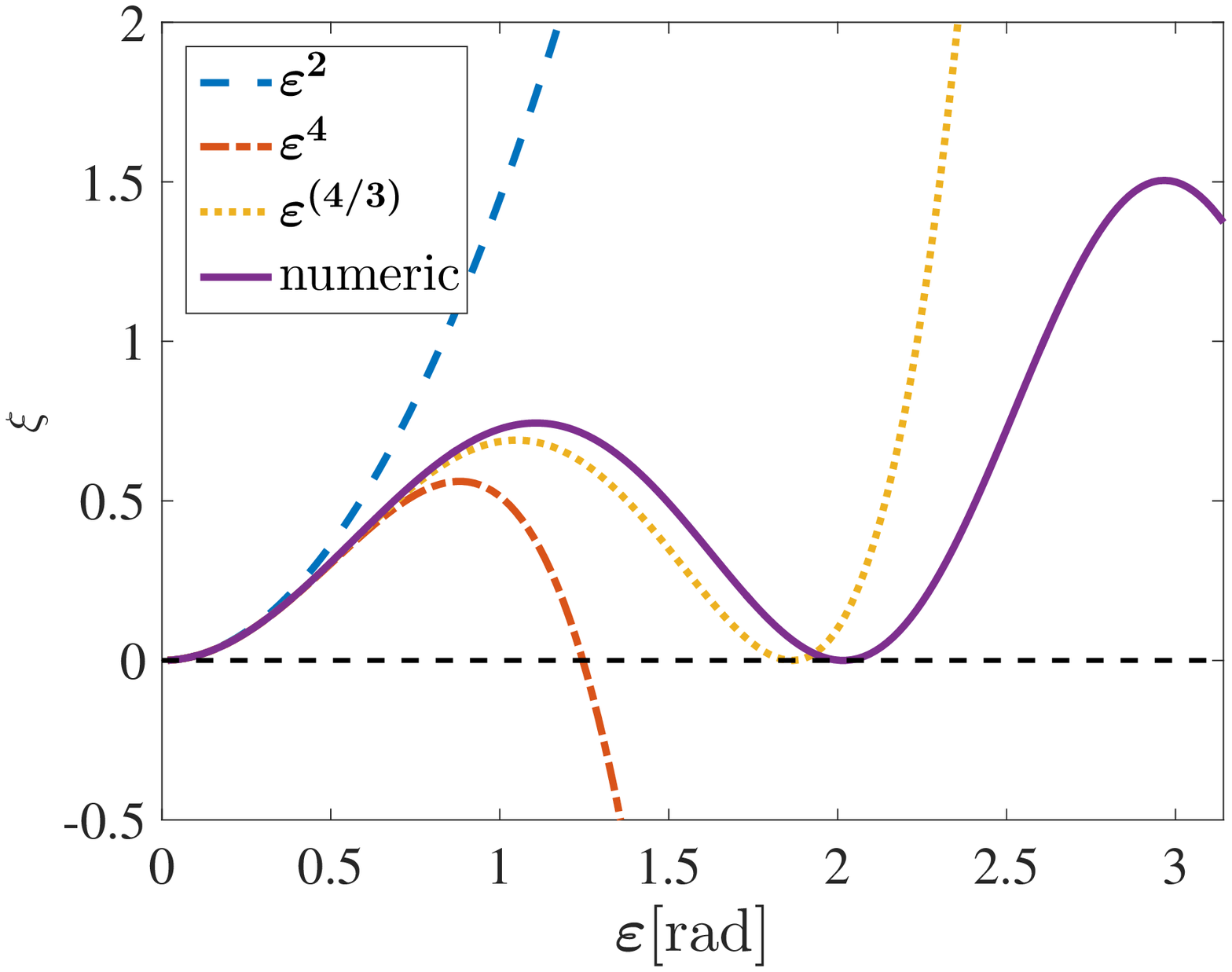}
    \caption{}
    \label{fig:xi_vs_eps_sq}
    \end{subfigure}
    \begin{subfigure}[h]{0.49\textwidth}
    \includegraphics[width=\textwidth]{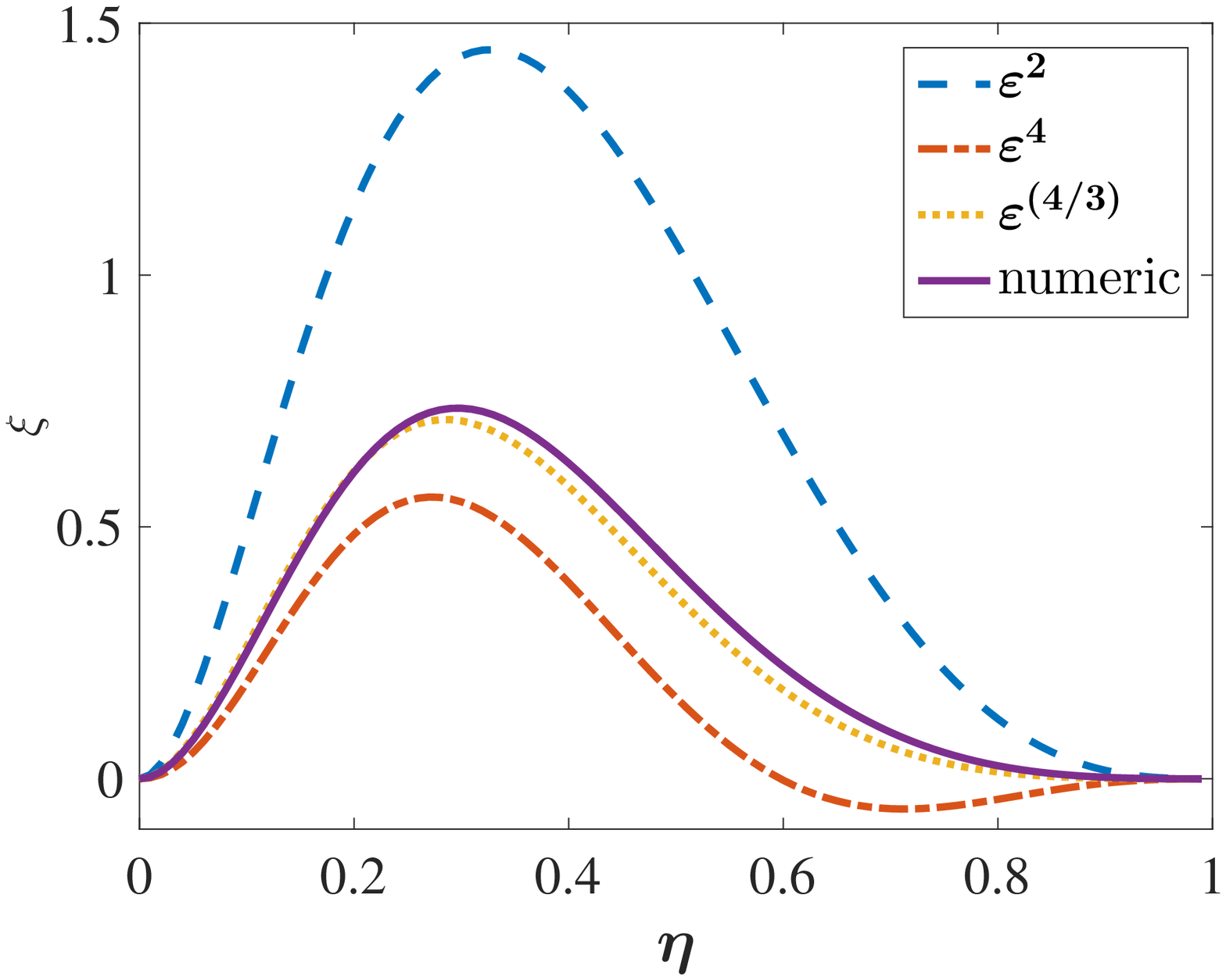}
    \caption{}
    \label{fig:xi_vs_eta_sq}
    \end{subfigure}
    \caption{Lighthill's energetic efficiency $\xi$ under the square gait: (a) Plot of $\xi$ vs. $\vep$ for $\eta=1/3$. (b) Plot of $\xi$ vs. $\eta$ for $\varepsilon=1$. Solid curves - exact (numeric) computation. Dashed, dash-dotted and dotted curves - approximations of $O(\vep^2)$, $O(\vep^4)$, and $O(\vep^{4/3})$, respectively.}
\end{figure}

\paragraph{Efficiency optimization - square gait:}

We now use the approximations of Lighthill's energetic efficiency in order to obtain optimal values of both length ratio $\eta$ and stroke amplitude $\vep$. Using only the leading-order approximation of the efficiency in \eqref{eq:xi2sq} gives an optimal length ratio of $\eta^*=0.327$. Taking also the next-order correction term into account and implementing the same calculation steps shown above in section 3\ref{sec:Optimal geometry for maximal displacement}, the optimal geometry and amplitude are obtained as:
\begin{equation}
\eta^*=0.2853 \quad \varepsilon^*=0.908\text{rad}=52^{\circ} \xi=0.5838
\end{equation}
Using $O(\varepsilon^{(4/3)})$ approximation, the optimal geometry and amplitude are obtained by solving a system of two polynomial equations as:
\begin{equation}
\eta^*=0.2754 \quad \varepsilon^*=1.1038\text{rad}=63.2^{\circ} \quad \xi=0.7294
\end{equation}
Finally, numerical calculation of $\xi$ under the square gait and conducting optimization using MATLAB's function {\bf fmincon}  gives optimal values of:

\begin{equation}
\eta^*=0.2879 \quad \varepsilon^*=1.133\text{rad}=64.9^{\circ} \quad \xi=0.7709
\end{equation}
Figure \ref{fig:xi_vs_eta_sq} shows a plot the approximations of $\xi$ as a function of $\eta$ for a large amplitude of $\varepsilon=1$, compared to the exact computation of $\xi$ obtained numerically. While there are large discrepancies in the value of the efficiency $\xi$ (best captured by the $O(\vep^{4/3})$ approximation), all approximations predict optimal values around  $\eta \approx 0.3$.

\subsection{Circular gait}
For the circular gait, the constant-power period time $T_p$ has been approximated only to first order $O(\vep)$ in \eqref{eq:cir_Time_constpower}. Using  (\ref{eq:disp_expansion}), the expansion of Lighthill's efficiency in \eqref{eq:eff_expansion} can thus be obtained only to leading order, as:
\begin{equation}
\xi=\vep^2 \dfrac{3(1-\eta)^3(\eta+3)}{c_t l E\left(\left.\frac{\pi}{2}\right|\frac{\eta^3+3\eta^2-3\eta-1}{\eta^2(\eta+3)}\right)}+O(\vep^3).
\label{xi2.cir} \end{equation}
Similar to \eqref{eq:xi43}, we define the $O(\vep^{4/1})$ approximation as:
\begin{equation}
\xi^{(4/1)}=\dfrac{\left(\varepsilon^2 X^{(2)+}\varepsilon^4X^{(4)} \right)^2}{\left(T_p^{(1)}\right)^2} \label{eq:xi41}
\end{equation}
Substituting \eqref{eq:disp_expansion} and \eqref{eq:cir_Time_constpower} into \eqref{eq:xi41} gives the expression for  the $O(\vep^{4/1})$ approximation (not shown for brevity). The efficiency $\xi$ and its approximations are plotted in Figure \ref{fig:xi_cir_vs_eps} as a function of the stroke amplitude $\vep$ under equal links lengths $\eta=1/3$, and in Figure \ref{fig:xi_cir_vs_eta} as a function of $\eta$ under amplitude of $\vep=1$. In both plots, $O(\vep^2)$ approximations appear in dashed curves, $O(\vep^{4/1})$ approximations appear in dotted curves, and the exact values computed numerically are shown in solid curves. As before, the leading-order $O(\vep^2)$ approximation does not capture the optimum with respect to amplitude but both approximations capture optimum with respect to $\eta$.

\paragraph{Efficiency optimization - circular gait:}
Using the leading-order approximation for $\xi$ in \eqref{xi2.cir} and numerically searching for maximum with respect to $\eta$, optimal geometry is obtained as $\eta^*=0.3139$, with an efficiency of $\xi^{(2)}=1.5565$. Using the $O(\vep^{4/1})$ approximation in \eqref{eq:xi41}, optimal values of both $\vep$ and $\eta$ are obtained numerically as: \begin{equation}
\eta^*=0.2379 \quad \varepsilon^*=1.3821\text{rad} \quad \xi=1.22 .
\end{equation}
Finally, numerical calculation of $\xi$ under the circular gait conducting optimization using MATLAB's {\bf fmincon} function give optimal values of:
\begin{equation}
\eta^*=0.2684 \quad \varepsilon^*=1.3747\text{rad} \quad \xi=1.2777 .
\end{equation}
Again, the $O(\vep^{4/1})$ expression in \eqref{eq:xi41} achieves a reasonable approximation of the optimum, as also seen from Figures \ref{fig:xi_cir_vs_eps} and \ref{fig:xi_cir_vs_eta}.

\begin{figure}[!t]
    \centering
    \begin{subfigure}[h]{0.49\textwidth}
    \includegraphics[width=\textwidth]{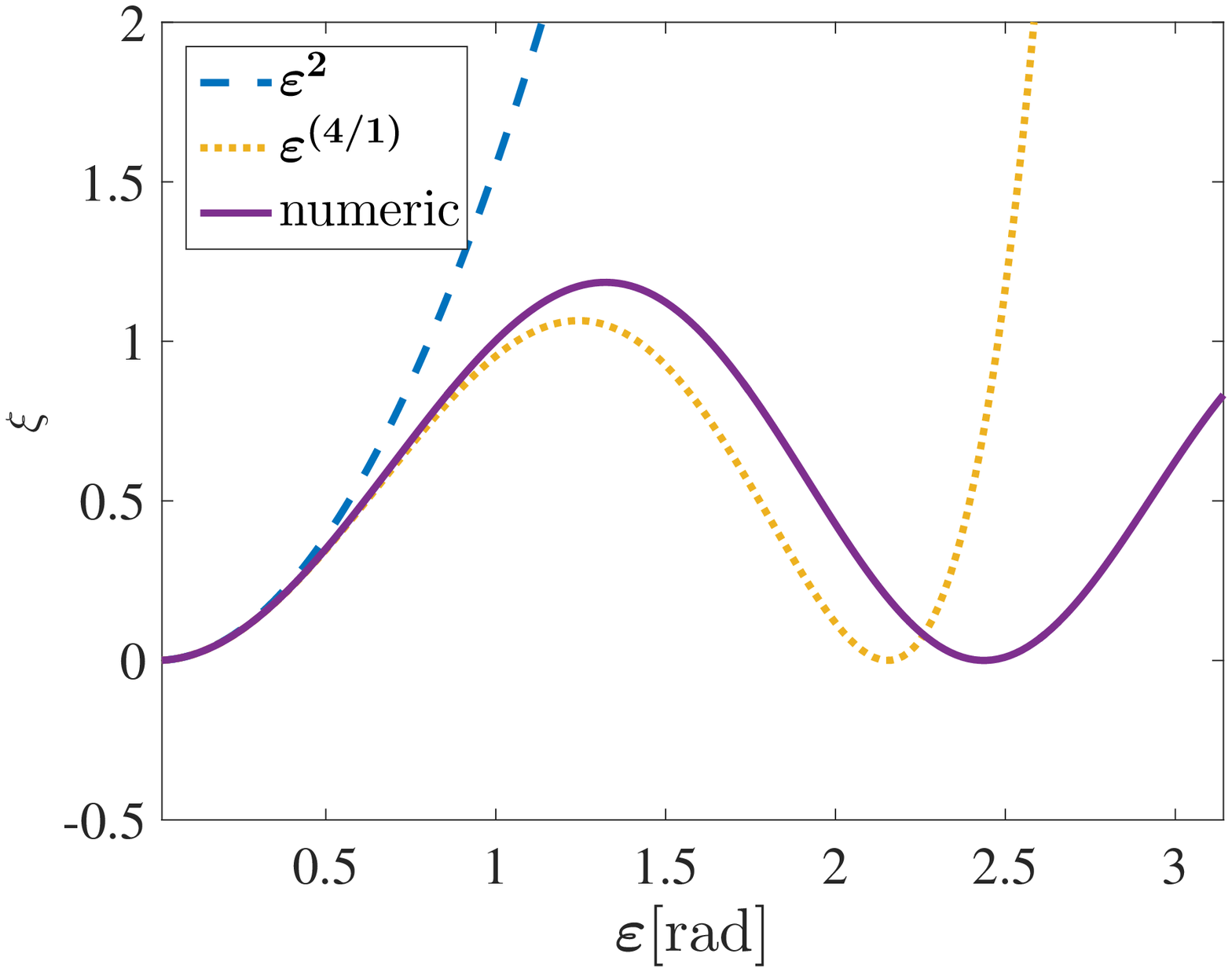}
    \caption{}
    \label{fig:xi_cir_vs_eps}
    \end{subfigure}
    \begin{subfigure}[h]{0.49\textwidth}
    \includegraphics[width=\textwidth]{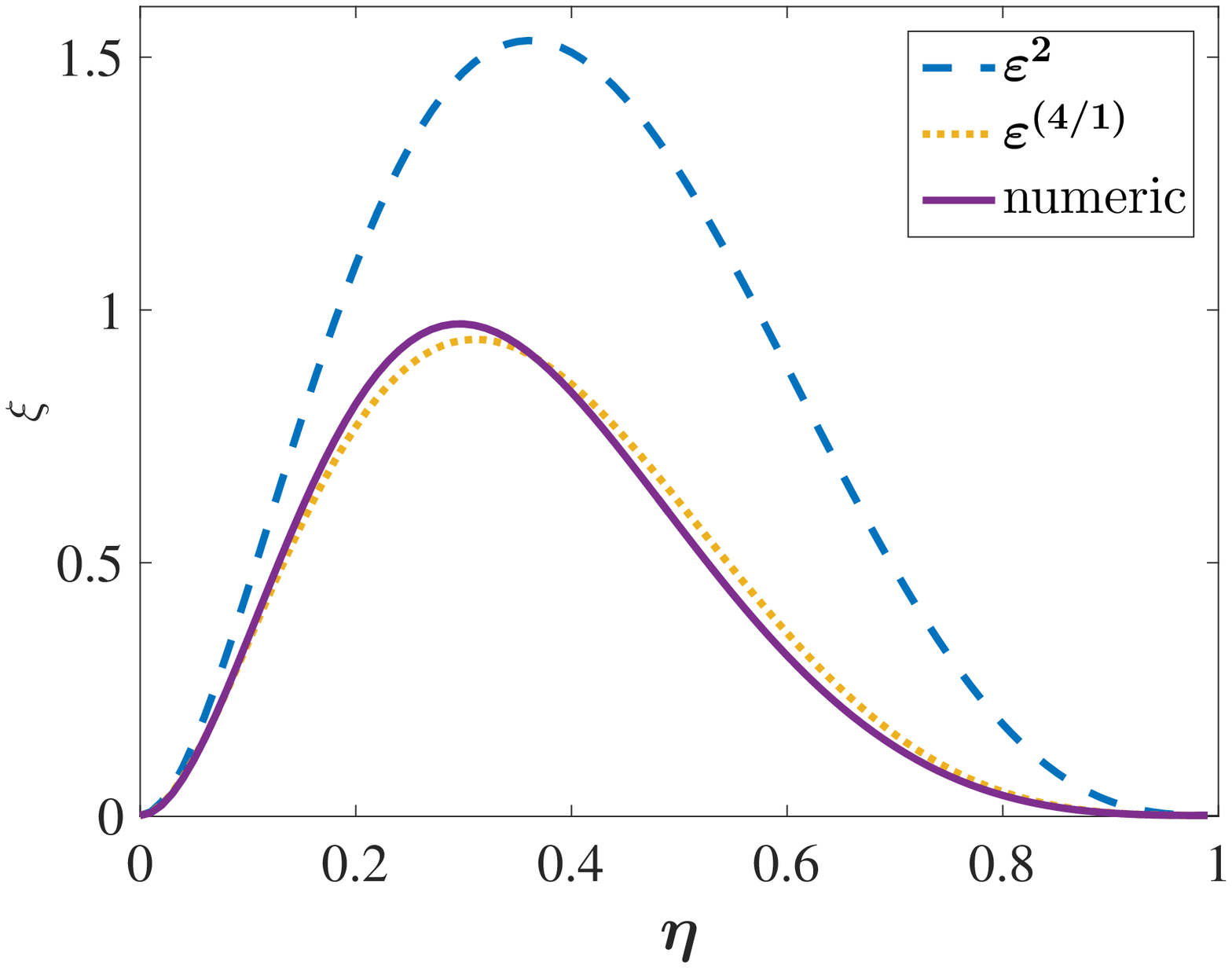}
    \caption{}
    \label{fig:xi_cir_vs_eta}
    \end{subfigure}

    \caption{Lighthill's energetic efficiency $\xi$ under the circular gait: (a) Plot of $\xi$ vs. $\vep$ for $\eta=1/3$. (b) Plot of $\xi$ vs. $\eta$ for $\varepsilon=1$. Solid curves - exact (numeric) computation. Dashed and dotted curves - approximations of $O(\vep^2)$ and $O(\vep^{4/1})$, respectively.}
    \label{fig:xi_cir}
    \vspace{-10pt}
\end{figure}

An important observation is that even in leading order $O(\vep^2)$ approximation of Lighthill's energetic efficiency $\xi$, the optimal geometry $\eta^*$ for maximizing $\xi$ actually depends on the shape of the chosen gait trajectory. This can clearly be seen in the different dependence of $\xi^{(2)}$ on $\eta$ in \eqref{eq:xi2sq} and \eqref{xi2.cir}. This is a substantial difference from the case of maximizing displacement $X$, where dependence of the leading-order approximation in $\eta$ is independent of the gait's shape, as manifested in Theorem \ref{th:global_eta}.

\section{Conclusion}
In this work we have analyzed the motion of the three-link "Purcell swimmer".
We provided a systematic method for deriving an expansion of the velocities of the swimmer and presented leading order expressions and next-order corrections for the net displacement over one period in the cases of a square and circular gait.
Examination of the correction terms confirms that there is a reversal in the direction of the net displacement at high amplitudes, a result which has previously only been shown numerically. The gait amplitude and swimmer geometry that optimize the displacement over one period were approximated using the first two terms in the expansion.
Additionally, by writing asymptotic expressions for the period time under constant power expenditure, we were able to write, for the first time, leading-order term for the energetic efficiency of the square and circular gaits as well as a next-order correction term for the square gait. Once again, we used the obtained expressions in order to find the energetically optimal gait amplitude and swimmer geometry. The results demonstrate the utility of perturbation methods for obtaining approximate explicit expressions for nonlinear dynamics of locomotion systems, which enable analysis and optimization of their performance.

We now briefly discuss some limitations of this work and list possible directions for future extension of the research. First, the swimmer's dynamic equations have been formulated using the simplification of resistive force theory \cite{cox1970,gray1955propulsion}, in which hydrodynamic interaction between the links is neglected. It is well-known (cf. \cite{avron&raz08,Alouges20131189}) that this assumption holds only for highly slender links, and for small stroke amplitudes where the gap between the links remains large even in the vicinity of the joints. Hydrodynamic interactions may be accounted for by using more refined models of slender body theory as in \cite{huber2011micro,tam2007optimal}. The decoupled relations in \eqref{eq:matrixrft} should then be modified to account for inter-link resistance, and the analysis will probably become purely numerical due to the added complexity of the dynamics. Nevertheless, the structure and geometric symmetries of the dynamics in \eqref{eq:dynamics} will be maintained without qualitative changes. In fact, according to the numerical investigation in \cite{huber2011micro,tam2007optimal}, no large quantitative changes in the results are expected even for moderate stroke amplitudes ($\vep \approx 2 \; rad$).

Second, the gait optimization conducted in this work was limited to varying the amplitude of predefined shape trajectories and did not address the possibility of other trajectories which may have better performance. Shape optimization of swimmers has been extensively studied, cf. \cite{Delfour2011,Moubachir2006,Walker2015}. One possible approach \cite{avron.optimal2004,shapere89b,tam2007optimal} is to represent a subset of all possible periodic shape changes by a finite set of variables (e.g. coefficients of a truncated Fourier series) and perform numerical optimization over this discrete set of variables. Another approach is optimal control theory \cite{bryson1975applied}, which is based on calculus of variations. This approach has been exploited by \cite{alouges2009optimal} for obtaining energy-optimal gaits of unidirectional axisymmetric swimmers and by \cite{giraldi2015optimal} for gait optimization of multi-link microswimmers under input constraints. In our recent  work \cite{oren.opt}, we have utilized optimal control in order to obtain unconstrained optimal gaits for maximal displacment of Purcell's swimmer, which exactly reproduce the gaits obtained numerically in \cite{tam2007optimal}. Furthermore, we are currently working on extending the results of \cite{alouges2009optimal} to planar motion of multi-link microswimmers.

Third, experiments conducted on a macro-scale swimmer prototype in a highly viscous fluid \cite{Gutman2015Symmetries} have recently demonstrated the symmetry properties of Purcell's swimmer performing certain symmetric gaits. Similar experiments can be done to verify some results from this paper, e.g. the dependency of net displacement on swimmer geometry and gait amplitude.
Finally, in many practical situations the actuation of robotic swimmers will likely be by applying torques at the joints rather than prescribing the joint angles directly \cite{or2012asymmetry}. This calls for the formulation of an asymptotic expansion of the relation between the joint torques and the joint angles.
 \vskip6pt

%
%
%
%
%



\bibliographystyle{IEEEtran}
\bibliography{mybibliography}

\newpage
\setcounter{page}{1}
{\noindent
\Large Supplementary Document to: "Optimization and small-amplitude analysis of Purcell's three-link microswimmer model"}

\appendix

\addcontentsline{toc}{section}{Appendix~\ref{app:A}:}
\section*{A. Elements of the dynamics matrix $\vecG$ and their derivatives evaluated at $\vecphi=0$}
\label{app:A}

\begin{table}[!h]
{\scriptsize
\begin{tabular}{|>{$}c<{$}|>{$}c<{$}|}
\hline\\
G_{11}(0)=G_{12}(0)=0
& \dfrac{\partial G_{12}}{\partial\phi_1}(0)=-\dfrac{\partial G_{11}}{\partial\phi_2}(0)=-\dfrac{l(1-\eta)^3(\eta+1)}{16} \\[12pt]
\dfrac{\partial G_{2i}}{\partial\phi_j}(0)=0, \quad i,j=1,2
& \dfrac{\partial^2G_{32}}{\partial\phi_2^2}(0)=-\dfrac{\partial^2G_{31}}{\partial\phi_1^2}(0)=\dfrac{3\eta\left(\eta^2-1\right)^2}{16} \\[12pt]
\dfrac{\partial G_{3j}}{\partial\phi_i}(0)=0, \quad i,j=1,2
&  \dfrac{\partial^2G_{32}}{\partial\phi_1^2}(0)=-\dfrac{\partial^2G_{31}}{\partial\phi_2^2}(0)= -\dfrac{3\eta(1-\eta)^3(\eta+3)}{16} \\[12pt]
\dfrac{\partial^2G_{1k}}{\partial\phi_i\partial\phi_j}(0)=0, \quad i,j,k=1,2
& \dfrac{\partial^2G_{21}}{\partial\phi_2^2}(0)=\dfrac{\partial^2G_{22}}{\partial\phi_1^2}(0)=\dfrac{l(3\eta+1)(1-\eta)^3}{32} \\[12pt]

\dfrac{\partial^3G_{2i}}{\partial\phi_j^2\partial\phi_k}(0)=0, \quad i,j,k=1,2
& \dfrac{\partial^2G_{21}}{\partial\phi_1\partial\phi_2}(0)=\dfrac{\partial^2G_{22}}{\partial\phi_1\partial\phi_2}(0)= -\dfrac{l(1-\eta)^3(\eta+1)}{32} \\[12pt]
G_{21}(0)=G_{22}(0)=-\dfrac{l(\eta-1)^2}{8}
& \dfrac{\partial^3G_{11}}{\partial\phi_1^3}(0)=-\dfrac{\partial^3G_{12}}{\partial\phi_2^3}(0)=\dfrac{l(3\eta-1)^2(\eta^2-1)^2}{64} \\[12pt]
G_{31}(0)\!=\!-G_{32}(0)\!=\!\dfrac{(\eta\!-\!1)^2(\eta\!+\!2)}{4}
& \dfrac{\partial^3G_{12}}{\partial\phi_1^3}(0)\!=-\dfrac{\partial^3G_{11}}{\partial\phi_2^3}(0)\!=\dfrac{l(1\!-\!\eta)^3(9\eta^3\!+\!27\eta^2\!-\!5\eta\!+\!1)}{64} \\[12pt]
\dfrac{\partial G_{11}}{\partial\phi_1}(0)\!=\!-\dfrac{\partial G_{12}}{\partial\phi_2}(0)\!=\!-\dfrac{l(\eta^2\!-\!1)^2}{16}
& \dfrac{\partial^3G_{11}}{\partial\phi_1^2\partial\phi_2}(0)\!=\!-\dfrac{\partial^3G_{12}}{\partial\phi_2^2\partial\phi_1}(0)\!=\!\dfrac{l(\eta\!-\!1)^3(\eta\!+\!1)(-3\eta^2\!+\!6\eta\!+\!1)}{64} \\[12pt]
\dfrac{\partial^2G_{21}}{\partial\phi_1^2}(0)=\dfrac{\partial^2G_{22}}{\partial\phi_1^2}(0)=\dfrac{l(\eta^2-1)^2}{32}
& \dfrac{\partial^3G_{12}}{\partial\phi_1^2\partial\phi_2}(0)\!=-\dfrac{\partial^3G_{11}}{\partial\phi_2^2\partial\phi_1}(0)\!=\dfrac{l(1\!-\!\eta)^3(3\eta^3\!+\!3\eta^2\!+\!9\eta\!+\!1)}{64} \\[12pt]
\hline
\end{tabular}}
\caption{Partial derivatives of elements of the matrix $\vecG$ at $\vecphi=0$}
\label{table:Gder}
\end{table}

\addcontentsline{toc}{section}{Appendix~\ref{app:B}:}
\section*{B. Proof of Proposition 3.1}
\label{app:B}
 \renewcommand{\theequation}{B-\arabic{equation}}
 \begin{prop}\label{th:x_disp_app}
For a symmetric, three linked "Purcell swimmer" performing a square or circular gait with amplitude $\varepsilon$, the leading-order term and next-order correction for the displacement $X$ over one full stroke in the direction of $x$ axis  are given as:
\begin{equation}
X=f_2(\eta)\varepsilon^2-f_4(\eta)\varepsilon^4+O(\varepsilon^6)
\label{eq:disp_expansion_app}
\end{equation}
With:
\begin{IEEEeqnarray}{l}
f_2(\eta)=C_2\eta l(1-\eta)^3(\eta+3)\nonumber\\
f_4(\eta)=C_4\eta l(1-\eta)^3(\eta^7+3\eta^6-10\eta^5-22\eta^4+29\eta^3+95\eta^2+44\eta+20)
\label{eq:disp_f2f4_app}
\end{IEEEeqnarray}
Where, for the square gait given in (\ref{eq:squaregait}) we have $C_2=1/4,C_4=1/192$, and for the circular gait given in (\ref{eq:circlegait}) $C_2=\pi/16,C_4=\pi/1024$.
\end{prop}

\begin{proof}
To find the displacement, we first integrate each of eq. (\ref{eq:thetadot1}),(\ref{eq:thetadot2}),(\ref{eq:thetadot3}) after substituting the gait $s_i(t)$, square (\ref{eq:squaregait}) or circular (\ref{eq:circlegait}). This gives the rotation angle
$\theta(t)$, which is substituted into (\ref{eq:velocities})
for obtaining the velocities $\dot x(t)$. Next,  the velocities are integrated in order to obtain the displacement and rotation of the swimmer, as detailed next.

\paragraph{Square gait:}
For the square gait, we focus on the first quarter of period, $t \in [0,2]$ only, and the rest of the period will be obtained later by symmetry considerations. From substitution of \eqref{eq:squaregait} into \eqref{eq:thetadot1}-\eqref{eq:thetadot3} and integrating under zero initial conditions $\theta^{(i)}(0)=0$, the rotation angle $\theta(t)$ is given by
\begin{IEEEeqnarray}{rCl}
\theta^{(1)}(t) &=& \dfrac{1}{4}(\eta-1)^2(\eta+2)t,\\
\theta^{(3)}(t) &=& \eta(\eta-1)^2\left[-\dfrac{1}{32}(\eta+1)^2t^3+\dfrac{3}{16}(\eta+1)t^2-\dfrac{3}{8}\eta t\right]\nonumber
\end{IEEEeqnarray}
where the second-order term of $\theta(t)$ is identically zero. Substituting into \eqref{eq:velocities}, the velocity $\dot x(t)$ is obtained as:

\begin{equation}
\resizebox{\textwidth}{!}{$
\begin{array}{lrcl}
\varepsilon^2: & \dot x^{(2)}(t) &=& \frac{1}{32}\eta l(\eta-1)^2(-\eta^2+2\eta+7)t+\frac{1}{8}\eta l(\eta+1)(\eta-1)^2,\\[8pt]
\varepsilon^4: & \dot x^{(4)}(t) &=& -\frac{1}{3072}\eta l(\eta\!-\!1)^2(\eta^8\!-\!6\eta^7\!-\!21\eta^6\!+\!36\eta^5\!+\!99\eta^4\!-\!150\eta^3\!-\!219\eta^2\!+\!88\eta\!+\!44)t^3\\[6pt]
&&&+\frac{1}{256}\eta l (1\!-\!\eta)^3(256\eta^7\!+\!\eta^6\!-\!6\eta^5\!-\!6\eta^4\!+\!15\eta^3\!+\!51\eta^2\!+\!6\eta\!-\!30)t^2)\\[6pt]
&&&+\frac{1}{32}\eta l(\eta\!-\!1)^2(2\eta^3\!-\!9\eta^2\!+\!3)t+\frac{1}{48}\eta l(\eta\!-\!1)^2(6\eta^2\!-\!5\eta\!+\!1)
\end{array}$}
\end{equation}
Integrating over $0 \leq t \leq 2$ then gives the net translation in $x$ for a quarter gait. Due to symmetries of the swimmer and the gait, this net translation is the same for every quarter of the gait, while motions in both $y$  and $\theta$ cancel out \cite{Gutman2015Symmetries}. Thus, multiplying by 4 the result for the quarter gait gives the full cycle net translations in $x$ as given in \eqref{eq:disp_expansion_app}. Higher order terms can be derived similarly, though writing the solution becomes more complicated since higher order expansions of $\theta(t)$ have to be calculated for integration of higher order terms of $X$.

\paragraph{Circular gait}

For the circular gait, integration with initial condition $\theta^{(i)}(0)=0$ for all orders yields:
\begin{IEEEeqnarray}{rCl}
\theta^{(1)}(t) &=& -\dfrac{\sqrt{2}}{4}(\eta-1)^2(\eta+2)\sin(t),\\
\theta^{(3)}(t) &=& -\dfrac{3}{64}\sqrt{2}\eta(\eta-1)^2(\eta^2+4\eta-3)\sin(t)-\dfrac{1}{64}\sqrt{2}\eta(\eta-1)^2(\eta^2-3)\sin(3t)\nonumber
\end{IEEEeqnarray}
As mentioned, the 2nd order expression for $\theta(t)$ is identically zero.
The velocity $\dot x(t)$ is obtained by substituting $\theta(t)$ into (\ref{eq:velocities}). Note that the first and third order terms in $\dot x$ are identically zero.
\begin{equation}
\resizebox{\textwidth}{!}{$
\begin{array}{lrcl}
\varepsilon^2: & \dot x^{(2)}(t) &=& -\frac{1}{16}\eta l(\eta-1)^2(\eta^2-2\eta-7)\sin^2(t)+\frac{1}{8}\eta l(\eta+1)(\eta-1)^2,\\[8pt]
\varepsilon^4: & \dot x^{(4)}(t) &=& -\frac{1}{1536}\eta l(\eta\!-\!1)^2(\eta^8\!-\!15\eta^6\!+\!63\eta^4\!-\!165\eta^2\!+\!160\eta\!-\!76)\cos(2t)\\[6pt]
&&&+\frac{1}{6144}\eta l (1\!-\!\eta)^3(-\eta^7\!+\!\eta^6\!+\!26\eta^5\!-\!10\eta^4\!-\!109\eta^3\!+\!89\eta^2\!+\!212\eta\!-\!180)\cos(4t)\\[6pt]
&&&+\frac{1}{2048}\eta l(1\!-\!\eta)^3(\eta^7\!+\!3\eta^6\!-\!10\eta^5\!-\!22\eta^4\!+\!29\eta^3\!+\!95\eta^2\!+\!44\eta\!+\!20)\IEEEyesnumber
\end{array}$}
\end{equation}
Finally, integrating over a full cycle $0 \leq t \leq 2\pi$, the net displacement over the full cycle $X$ is obtained as in \eqref{eq:disp_expansion_app}.
\end{proof}

\addcontentsline{toc}{section}{Appendix~\ref{app:C}:}
\section*{C. Leading-order dependency on $\eta$}
\label{app:C}
 \renewcommand{\theequation}{C-\arabic{equation}}
  \setcounter{equation}{0}
Here we prove that the leading-order expression for the $x$ displacement of any periodic gait can be decomposed into the product of two parts, one depending only on the swimmer's geometry and the other depending only on the shape of gait trajectory. The result of this is that to leading order, the displacement of a given swimmer is the same for any gait up to a constant, and so we can find a swimmer geometry which optimizes the displacement for all gaits.


\begin{theorem*}
Consider a swimmer with two shape variables whose dynamics can be written in the form of equation \eqref{eq:gauge_sym}, under a small-amplitude gait $\vecphi(t)=\vep \vecs(t)$. The leading-order approximation of the swimmer's displacement $X$ over one period can be written in the form: $X^{(2)}=C(\vecs)   \cdot f(\eta)$, where $C(\vecs)$ depends on the shape of the unscaled trajectory and  $f(\eta)$ depends on the swimmer's geometric structure.
\end{theorem*}


\begin{proof}
From (\ref{eq:velocities}) we have that the swimmer velocities are:
\begin{equation}
\vecqd^{(2)} = \left[\theta^{(1)}\vecJ\vecG(0)+ \left(s_1\dfrac{\partial}{\partial\phi_1}\Big|_{(0,0)}+s_2\dfrac{\partial}{\partial\phi_2}\Big|_{(0,0)}\right)\vecG\right]\dot{\vecs}(t)
\label{eqapp1} \end{equation}
and from (\ref{eq:thetadot1}) we have the angular velocity:
\begin{equation}
\thetad^{(1)}(t)=\sum_{j=1}^2\vecG_{3j}(0) \dot s_j(t)
\end{equation}
integrating we get the rotation angle:
\begin{equation}
\theta^{(1)}(t)=\int_{t_0}^{t} \thetad^{(1)}dt=\sum_{j=1}^2\vecG_{3j}(0)\int_{t_0}^{t} \dot s_j(t) dt= \sum_{j=1}^2\vecG_{3j}(0)\left(s_j(t)-s_j(t_0)\right)
\end{equation}
substituting into \eqref{eqapp1} we obtain the velocities:
\begin{eqnarray}
\vecqd^{(2)} &=& \left[\vecJ\vecG(0)(G_{31}s_1(t)+G_{32}s_2(t))+\dfrac{\partial\vecG}{\partial\phi_1}s_1(t)+\dfrac{\partial\vecG}{\partial\phi_2}s_2(t)\right]\dot{\vecs}\nonumber\\
&&-\left(\vecJ\vecG(0)(G_{31}s_1(0)+G_{32}s_2(0))\right)\dot{\vecs}
\end{eqnarray}
The velocity in the direction of the $x$ axis is (all terms and derivatives of the matrix $\vecG$ are evaluated at $\vecphi=0$, the notation has been dropped for convenience):
\begin{eqnarray}
\dot x^{(2)}&=&\left(-G_{21}G_{31}+\dfrac{\partial G_{11}}{\partial \phi_1}\right)s_1(t)\dot{s}_1(t)+\left(-G_{22}G_{31}+\dfrac{\partial G_{12}}{\partial \phi_1}\right)s_1(t)\dot{s}_2(t)\nonumber\\[6pt]
&&+\left(-G_{21}G_{32}+\dfrac{\partial G_{11}}{\partial \phi_2}\right)s_2(t)\dot{s}_1(t)+\left(-G_{22}G_{32}+\dfrac{\partial G_{12}}{\partial \phi_2}\right)s_2(t)\dot{s}_2(t)\nonumber\\[10pt]
&&+\left(G_{31}s_1(0)+G_{32}s_2(0)\right)\left(G_{21}\dot{s}_1(t)+G_{22}\dot{s}_2(t)\right)\nonumber\\[10pt]
&=& C_{11}s_1\dot s_1+C_{12}s_1\dot s_2+C_{21}s_2\dot s_1+C_{22}s_2\dot s_2+C_{01}\dot s_1+C_{02}\dot s_2
\end{eqnarray}
Integrating over the full period to get the net displacement:
\begin{eqnarray}
X^{(2)}&=&\int_0^T \dot x dt \nonumber \\
&=&C_{11}\int_0^Ts_1\dot{s}_1dt+C_{12}\int_0^T s_1\dot{s}_2 dt+C_{21}\int_0^T s_2\dot{s}_1 dt \nonumber \\
&&+C_{22}\int_0^T s_2\dot{s}_2dt+C_{01}\int_0^T \dot{s}_1 dt+C_{02}\int_0^T \dot{s}_2dt
\end{eqnarray}
and using integration by parts
\begin{eqnarray}
X^{(2)}&=&
0.5 C_{11}s_1^2\bigg|_0^T+C_{12}\int_0^T s_1\dot{s}_2 dt+C_{21}\left(s_1 s_2\bigg|_0^T - \int_0^T s_1\dot{s}_2 dt\right)\nonumber\\
&&+0.5 C_{22}s_2^2\bigg|_0^T+C_{01}s_1\bigg|_0^T+C_{02}s_2\bigg|_0^T
\end{eqnarray}
since we are dealing only with periodic gaits, we have $\vecs(0)=\vecs(t)$ and we are left with:
\begin{equation}
X^{(2)} = \left(C_{12}-C_{21}\right) \int_0^Ts_1\dot{s}_2 dt=\left(-G_{22}G_{31}+\dfrac{\partial G_{12}}{\partial \phi_1}+G_{21}G_{32}-\dfrac{\partial G_{11}}{\partial \phi_2}\right) \int_0^Ts_1\dot{s}_2 dt . \label{eq:int.app}
\end{equation}
Only the integral in \eqref{eq:int.app} depends on the gait's trajectory $\vecs(t)$, while the terms in parentheses depend only on the swimmer geometry $\eta$.
\begin{remark}
This proof still holds for non-smooth gaits such as the square gait where $\dot \vecs(t)$ is only piecewise continuous, since the set of discontinuity points is of measure zero and therefore they are Riemann integrable.
\end{remark}
\end{proof}

\addcontentsline{toc}{section}{Appendix~\ref{app:D}:}
\section*{D. Elliptic integrals}
\label{app:D}

The elliptic integral of the second kind is defined by:
\begin{equation*}
E(\varphi \setminus \alpha)=\int_0^\varphi\sqrt{ 1-\sin^2 \alpha \sin^2 \theta}d\theta,
\end{equation*}
or
\begin{equation*}
E[\varphi | m]=\int_0^\varphi\sqrt{ 1-m \sin^2 \theta}d\theta,
\end{equation*}
with $0<m=\sin^2 \alpha<1$.

For the case $\varphi=\pi/2$ we have the complete elliptic integral of the second kind:
\begin{equation*}
E(m)=\int_0^{\pi /2} \sqrt{ 1-m \sin^2 \theta}d\theta,
\end{equation*}

\addcontentsline{toc}{section}{Appendix~\ref{app:E}:}
\section*{E. Proof of Proposition 3.2}
\label{app:E}
 \renewcommand{\theequation}{E-\arabic{equation}}
  \setcounter{equation}{0}
  \begin{prop}
For a symmetric, three linked "Purcell swimmer" performing a square or circular gait with amplitude $\varepsilon$, the period times of one full stroke with constant power of $P_o=1$ exerted by the joints are expanded as follows.
\newline
Square gait:

\begin{IEEEeqnarray}{rcl}
T_{p} &=&\frac{\sqrt{6}}{3}\sqrt{\ct l^3 {\left({1 - \eta}^2\right)}^3}\:\varepsilon \nonumber\\&&+ \frac{\sqrt{6  \ct l^3}  {\left( 1-\eta \right)}^4 \left(\eta+1\right) \left(\eta^3 + 3\eta^2 - 3\eta + 1\right)}{12  \sqrt{{\left(1 - \eta^2\right)}^3}}\varepsilon^3+O(\varepsilon^5)
\label{eq:sq_Time_constpower_app}
\end{IEEEeqnarray}
Circular gait:
\begin{IEEEeqnarray}{l}
T_p=E\left(\frac{\eta^3+3\eta^2-3\eta-1}{\eta^2(\eta+3)}\right)\sqrt{\frac{c_tl^3}{3}\eta^2(1-\eta)^3(\eta+3)}\:\varepsilon+O(\varepsilon^3)
\label{eq:cir_Time_constpower_app}
\end{IEEEeqnarray}
\end{prop}

The function $E(\cdot)$ in \eqref{eq:cir_Time_constpower_app} denotes a complete elliptic integral of the second kind, whose definition is reviewed in above.
\begin{proof}
For each gait trajectory we define a new parametrization as $\vecphi(\sigma)$.  The time-parametrization $\sigma(t)$ that keeps constant power along the gait satisfies \eqref{eq:sigmaP}, and the period time is then calculated by expansion of \eqref{eq:Time_const_power}. 
For the square gait, as before, we consider only a quarter of the gait and then multiply by four due to symmetry.
Parametrization for the first quarter of the gait is
\begin{IEEEeqnarray}{lcl}
	\phi_1=\varepsilon, &\quad& \phid_1=0\nonumber\\
	\phi_2=-\sigma, && \phid_2=-\dot{\sigma}\\
	\sigma \in[-\varepsilon,\varepsilon].\nonumber
\end{IEEEeqnarray}
Using this in equation (\ref{eq:Time_const_power}) with $F(\sigma)=W_{22}(\vecphi(\sigma))$ and $P_o=1$ yields
\begin{equation}
	T_p=4 \int_{-\varepsilon}^\varepsilon\sqrt{W_{22}(\vecphi(\sigma))}d \sigma \label{eq:Tp.int.sq}
\end{equation}
using Taylor expansion of the integrand in \eqref{eq:Tp.int.sq},
the first two leading-order terms of $T_p$ are obtained as:
\begin{IEEEeqnarray}{rrcl}
\varepsilon: \,& T_{p}^{(1)}\, &=&\, \frac{\sqrt{6}}{3}\sqrt{\ct\, l^3 {\left({1 - \eta}^2\right)}^3}\,,\nonumber\\[12pt]
\varepsilon^3: \,& T_{p}^{(3)}\, &=&\, \frac{\sqrt{6 \, \ct \, l^3} \, {\left( 1-\eta \right)}^4 \left(\eta+1\right) \left(\eta^3 + 3\eta^2 - 3\eta + 1\right)}{12  \sqrt{{\left(1 - \eta^2\right)}^3}}
\end{IEEEeqnarray}
Which gives the expansion in \eqref{eq:sq_Time_constpower_app}.

For the circular gait, only first-order expression for the period time is obtained, due to complexity of the expressions. The parametrization used for this gait is
\begin{IEEEeqnarray}{lr}
	\phi_1=\varepsilon \cos(\sigma),\quad & \phid_1=-\varepsilon \sin(\sigma)\dot{\sigma}\nonumber\\
	\phi_2=\varepsilon\sin(\sigma), & \phid_2=\varepsilon\cos(\sigma)\dot{\sigma}\\
	\sigma\in [-\pi/4,\pi/4].\nonumber
\end{IEEEeqnarray}
As before, only a quarter of the gait is considered. Using \eqref{eq:Time_const_power}, one obtains:
\begin{equation}
	T_p=4 \int_{-\pi/4}^{\pi/4}\sqrt{\vctwo{-\sin(\sigma)}{\cos(\sigma)}^T\vecW(\vecphi(\sigma))\vctwo{-\sin(\sigma)}{\cos(\sigma)}\varepsilon^2}d \sigma \label{eq:Tint}
\end{equation}
The leading-order term is obtained by evaluating $\vecW$ at $\vecphi=0$ in \eqref{eq:Tint}, which yields:
\begin{equation}
	T_p^{(1)} = \int_{-\pi/4}^{\pi/4}\sqrt{\frac{\ct\, l^3(\eta-1)^3}{6}\left[ ( \eta + 1 )^3\, \sin(2\sigma)+(1-\eta)(\eta^2 + 4\eta + 1)\right]}
d \sigma \label{eq:Tp1}
\end{equation}
This integral is of the following form:
\begin{IEEEeqnarray}{rCl}
	\int_{-\frac{\pi}{4}}^\frac{\pi}{4}\sqrt{A+B\sin(2\sigma)} ds &=& \left.-\dfrac{E\left[\left.\frac{\pi}{4}-\sigma\right|\frac{2B}{A+B}\right]\sqrt{A+B\sin(2\sigma)}}{\sqrt{\frac{A+B\sin(2\sigma)}{A+B}}}\right|_{\sigma=-\pi/4}^{\pi/4}\nonumber\\
	&=&-\left. E\left[\left.\frac{\pi}{4}-\sigma\right|\frac{2B}{A+B}\right]\sqrt{A+B}\right|_{\sigma=-\pi/4}^{\pi/4}\nonumber\\
	&=& \left(E\left[\left.\frac{\pi}{2}\right|\frac{2B}{A+B}\right]-E\left[0\Big|\frac{2 B}{A+B}\Big.\right]\right)\sqrt{A+B}\nonumber\\
	&=& E\left(\frac{2B}{A+B}\right)\sqrt{A+B} \label{eq:elint}
\end{IEEEeqnarray}
where
\begin{equation}
A=\frac{\ct\, l^3(\eta-1)^3}{6}(1-\eta)(\eta^2 + 4\eta + 1), \quad B=\frac{\ct\, l^3(\eta-1)^3}{6}(\eta+1)^3
\end{equation}
and $E[\varphi|m]$ and $E(m)$ are the notations for the incomplete and complete elliptic integral of the second kind, respectively.
Using the notation of elliptic integral, the first-order term of $T_p$ is then obtained as in \eqref{eq:cir_Time_constpower_app}.
\end{proof}

%


\end{document}